\documentclass[]{interact}

\usepackage{epstopdf}% To incorporate .eps illustrations using PDFLaTeX, etc.
\usepackage[caption=false]{subfig}% Support for small, `sub' figures and tables

\usepackage[numbers,sort&compress]{natbib}% Citation support using natbib.sty
\bibpunct[, ]{[}{]}{,}{n}{,}{,}% Citation support using natbib.sty
\makeatletter% @ becomes a letter
\def\NAT@def@citea{\def\@citea{\NAT@separator}}% Suppress spaces between citations using natbib.sty
\makeatother% @ becomes a symbol again

\usepackage{amsmath}
\usepackage{amssymb}
\usepackage{amsfonts}
\usepackage{amsthm}
\usepackage{bm}
\usepackage{float}
\usepackage{tabularx}% http://ctan.org/pkg/tabularx
\usepackage{graphicx} % Required for including images in figures
\usepackage{booktabs} % Required for horizontal rules in tables
\usepackage{multicol} % Required for creating multiple columns in slides
\usepackage{enumitem}
\usepackage{multirow}

\DeclareMathOperator*{\argmax}{arg\,max}

\theoremstyle{plain}
\newtheorem{theorem}{Theorem}[section]
\newtheorem{lemma}{Lemma}[section]
\newtheorem{corollary}{Corollary}[theorem]
\newtheorem{proposition}{Proposition}[section]
\newtheorem{remark}{Remark}[section]

\newcommand{\numbereqn}{\addtocounter{equation}{1}\tag{\theequation}} % use \numberthis to add number in align* mode

\begin{document}

% \articletype{ARTICLE TEMPLATE}% Specify the article type or omit as appropriate

\title{Large-Scale Multiple Hypothesis Testing with the Normal-Beta Prime Prior}

\author{
\name{Ray Bai \textsuperscript{a}\thanks{CONTACT Ray Bai. Email: Ray.Bai@pennmedicine.upenn.edu} and Malay Ghosh \textsuperscript{b}}
\affil{\textsuperscript{a} Department of Biostatistics, Epidemiology, and Informatics,  University of Pennsylvania, Philadelphia, PA 19104 USA; \textsuperscript{b} Department of Statistics, University of Florida, Gainesville, FL 32611 USA}
}

\maketitle

\begin{abstract}
We revisit the problem of simultaneously testing the means of $n$ independent normal observations under sparsity. We take a Bayesian approach to this problem by studying a scale-mixture prior known as the normal-beta prime (NBP) prior. To detect signals, we propose a hypothesis test based on thresholding the posterior shrinkage weight under the NBP prior. Taking the loss function to be the expected number of misclassified tests, we show that our test procedure asymptotically attains the optimal Bayes risk when the signal proportion $p$ is known. When $p$ is unknown, we introduce an empirical Bayes variant of our test which also asymptotically attains the Bayes Oracle risk in the entire range of sparsity parameters $p \propto n^{-\epsilon}, \epsilon \in (0, 1)$. Finally, we also consider restricted marginal maximum likelihood (REML) and hierarchical Bayes approaches for estimating a key hyperparameter in the NBP prior and examine multiple testing under these frameworks.
\end{abstract}


\section{Introduction}

\subsection{Large-Scale Testing of Normal Means} \label{background}
\noindent
Suppose we observe an $n$-component random observation $(X_1, \ldots, X_n) \in \mathbb{R}^n$, such that
\begin{equation} \label{X=theta+eps}
X_i \sim \mathcal{N}(\theta_i, 1), \hspace{.2cm} i = 1, \ldots, n.
\end{equation}
This simple framework is the basis for a number of high-dimensional problems, including genetics, wavelet analysis, and image reconstruction \cite{JohnstoneSilverman2004}. Under model (\ref{X=theta+eps}), we are primarily interested in identifying the few signals ($\theta_i \neq 0$). This amounts to performing $n$ simultaneous tests, $H_{0i}: \theta_i = 0$ vs. $H_{1i}: \theta_i \neq 0$, $i = 1, \ldots, n$.  

In the high-dimensional setting where $n$ is very large, sparsity is a very common phenomenon. In genetics, for example, the $X_i$'s may represent thousands of gene expression data points, but only a few genes are significantly associated with the phenotype of interest. For instance, \cite{WellcomeTrust2007} has confirmed that only seven genes have a non-negligible association with Type I diabetes.

\subsection{Scale-Mixture Shrinkage Priors}
Scale-mixture shrinkage priors are widely used for obtaining (nearly) sparse estimates of $\bm{\theta}$ in (\ref{X=theta+eps}). These priors  take the form,
\begin{equation} \label{scalemixture}
\theta_i | \sigma_i^2 \sim \mathcal{N}(0, \sigma_i^2), \hspace{.2cm} \sigma_i^2 \sim \pi(\sigma_i^2), \hspace{.2cm} i = 1, \ldots, n,
\end{equation}
where $\pi: [0, \infty) \rightarrow [0, \infty)$ is a density on the positive reals. These priors typically contain heavy mass around zero, so that the posterior density is heavily concentrated around $\bm{0} \in \mathbb{R}^n$.  However, they also retain heavy enough tails in order to correctly identify and prevent overshrinkage of the true signals. Examples of (\ref{scalemixture}) include the popular horseshoe prior \cite{CarvalhoPolsonScott2010} and the Bayesian lasso \cite{ParkCasella2008}. Priors of the type (\ref{scalemixture}) have also been considered by numerous other authors: see, e.g. \cite{Strawderman1971, Berger1980, GriffinBrown2013, BhattacharyaPatiPillaiDunson2015, ArmaganDunsonLee2013, BhadraDattaPolsonWillard2017}.

 From (\ref{scalemixture}), we see that the posterior mean of $\theta_i$ under these priors is given by
\begin{equation} \label{posteriormean}
\mathbb{E} \{ \mathbb{E} ( \theta_i | X_i, \sigma_i^2 ) \} = \left\{ \mathbb{E} ( 1 - \kappa_i ) | X_1, \ldots, X_n \right\} X_i, 
\end{equation}
where $\kappa_i = 1/(1 + \sigma_i^2)$. By (\ref{posteriormean}), it is clear that the shrinkage weight $\kappa_i$ plays a crucial role in the amount of posterior shrinkage under these priors. 

\subsection{Multiple Testing Under Sparsity}

Assuming that the true data-generating model is a two-components mixture density, \cite{BogdanChakrabartiFrommletGhosh2011} studied the risk properties of a large number of multiple testing rules. Specifically, \cite{BogdanChakrabartiFrommletGhosh2011} considered a symmetric 0-1 loss function taken to be the expected total number of misclassified tests. Under mild conditions, \cite{BogdanChakrabartiFrommletGhosh2011} arrived at a simple closed form for the asymptotic Bayes risk under this loss. They termed this as the asymptotically Bayes optimal risk under sparsity (ABOS), or the Bayes Oracle risk. They then provided necessary and sufficient conditions for which a number of classical multiple test procedures (e.g. the Bonferroni correction or the Benjamini-Hochberg \cite{BenjaminiHochberg1995} procedure) could asymptotically equal the Bayes Oracle risk. A thorough discussion of this decision theoretic framework is presented in Section \ref{ABOS}.

Testing rules induced by scale-mixture shrinkage priors have also been studied within this decision theoretic framework. Since scale-mixture shrinkage priors of the form (\ref{scalemixture}) are absolutely continuous, they place zero mass at exactly zero. Thus, in order to classify means as either signal or noise, some thresholding rule must be applied. One method of doing this is by thresholding the posterior shrinkage weight $\kappa_i$ in (\ref{posteriormean}) as follows. For the $i$th component, the test procedure based on $\kappa_i$ is:
\begin{equation} \label{genericthreshold}
\textrm{Reject } H_{0i} \textrm{ if } \mathbb{E}(1-\kappa_i | X_1, \ldots, X_n) > \frac{1}{2}.
\end{equation}
Depending on how conservative the test must be, the fraction $1/2$ can be replaced by any $\alpha \in (0, 1)$, and then the final results will depend on $\alpha$. However, for most practical applications, it seems as though this `half-thresholding' rule of 1/2 is sensible \cite{CarvalhoPolsonScott2010, DattaGhosh2013, GhoshTangGhoshChakrabarti2016}.

Assuming that the $\theta_i$'s come from a two-components model, \cite{DattaGhosh2013} showed that rule (\ref{genericthreshold}) under the horseshoe prior asymptotically attains the Bayes Oracle risk up to a multiplicative constant.  \cite{GhoshTangGhoshChakrabarti2016} generalized this result to a general class of shrinkage priors of the form,
\begin{equation} \label{globallocal}
\theta_i | \tau, \lambda_i \sim \mathcal{N}(0,  \lambda_i \tau), \hspace{.2cm} \lambda_i \sim \pi(\lambda_i) = K \lambda_i^{-a-1} L(\lambda_i),
\end{equation}
where $\tau > 0$ is a variance rescaling parameter, $K$ is the constant of proportionality, $a > 0$, and $L(\cdot)$ is a measurable, nonconstant, slowly varying function. \cite{GhoshChakrabarti2017} later showed that thresholding rule (\ref{genericthreshold}) for this same class of priors (\ref{globallocal}) could even asymptotically attain the exact Bayes Oracle risk.  \cite{BhadraDattaPolsonWillard2017} also extended the same rule for the horseshoe+ prior, showing that rule (\ref{genericthreshold}) based on the horseshoe+ prior asymptotically attains the Bayes Oracle risk up to a multiplicative constant. 

Recently, \cite{Salomond2017} studied testing rule (\ref{genericthreshold}) under an even broader class of normal scale-mixture shrinkage priors (\ref{scalemixture}) which subsumes priors of the form (\ref{globallocal}). In this class, the prior on the scale parameter $\sigma_i^2$, $\pi(\sigma_i^2$), satisfies the three properties given in \cite{VanDerPasSalomondSchmidtHieber2016}. The properties in \cite{VanDerPasSalomondSchmidtHieber2016} are sufficient for scale-mixture priors to obtain the minimax posterior contraction rate under the sparse normal means model (\ref{X=theta+eps}). For priors satisfying these conditions, \cite{Salomond2017} derived upper bounds on the asymptotic Bayes risk for both non-adaptive and data-adaptive testing rules. He showed that the upper bound on the Bayes risk for this general class of priors is of the same order as the Bayes Oracle risk up to a multiplicative constant. 

The results in this manuscript were developed independently of \cite{Salomond2017} and give sharper bounds than those of \cite{Salomond2017}.  \cite{Salomond2017} did not obtain the exact asymptotic Bayes Oracle risk nor did he derive asymptotic lower bounds on the Type I and Type II errors or the Bayes risk. In contrast, our paper establishes tight upper \textit{and} lower bounds. To further highlight the distinction, we refer to testing rules as having the Bayes Oracle property if and only if they can be shown to asymptotically obtain the exact Bayes Oracle risk in \cite{BogdanChakrabartiFrommletGhosh2011}.  Further, the prior that we propose in this paper departs from the family of priors (\ref{globallocal}) considered by \cite{GhoshChakrabarti2017} because it does not require a variance rescaling parameter $\tau > 0$. Therefore, our results also do not automatically follow from those of \cite{GhoshChakrabarti2017}.

In this article, we consider a Bayesian scale-mixture shrinkage prior with the beta prime density as its scale parameter and no variance rescaling parameter $\tau$. We call our model the normal-beta prime (NBP) model. We highlight some of our contributions:

\begin{enumerate}
\item
We investigate the properties of the NBP model with \textit{varying} hyperparameters $(a, b)$.  Since we allow the hyperparameters to vary with the sample size, the concentration inequalities for the beta prime hierarchical model established in Section \ref{NBPPrior} are new, and thus, may be of independent interest for Bayesian inference involving the beta prime density as a prior. 

\item
We derive both lower and upper bounds on Type I and Type II probabilities under thresholding rules based on the NBP's posterior shrinkage factor. We show that with appropriate choices of $(a, b)$, our method asymptotically achieves the Bayes Oracle risk \textit{exactly}, both when the true number of signals $p$ is known and when it is unknown but is estimated with an appropriate empirical Bayes estimator. 

\item
Inspired by the recent work of \cite{VanDerPasSzaboVanDerVaart2017}, we introduce two other data-adaptive methods for estimating the hyperparameter $a$ in the NBP model based on restricted marginal maximum likelihood (REML) and hierarchical Bayes estimation. We study multiple testing procedures under these methods for a variety of shrinkage priors and show that they mimic oracle performance. 

\end{enumerate}
The organization of this paper is as follows. In Section \ref{NBPPrior}, we introduce the normal-beta prime (NBP) prior and establish new concentration inequalities for the beta prime density when it is employed as a scale parameter in Bayesian hierarchical models. In Section \ref{NBPTesting}, we consider two different testing rules -- one non-adaptive and one data-adaptive -- based on thresholding the posterior shrinkage weight and illustrate that they both possess the Bayes Oracle property. In Section \ref{DataAdaptiveMethods}, we introduce a restricted marginal maximum likelihood approach and a hierarchical Bayes approach for estimating the sparsity parameter in the NBP prior. In Section \ref{Simulations}, we present simulation results to validate our theoretical findings. Finally, in Section \ref{DataAnalysis}, we utilize the NBP prior to analyse a prostate cancer data set. 

Proofs for the propositions and theorems in this article are available in the Supplementary Materials.

\subsection{Notation}
We use the following notations for the rest of the paper. Let $\{ a_n \}$ and $\{ b_n \}$ be two  non-negative sequences of real numbers indexed by $n$, where $b_n \neq 0$ for sufficiently large $n$.  If $\lim_{n \rightarrow \infty} a_n/b_n = 1$, we write $a_n \sim b_n$. If $| a_n/b_n | \leq M$ for all sufficiently large $n$ where $M > 0$ is a positive constant independent of $n$, then we write $a_n = O(b_n)$. If $\lim_{n \rightarrow \infty} a_n/b_n = 0$, we write $a_n = o(b_n)$. Thus, $a_n = o(1)$ if $\lim_{n \rightarrow \infty} a_n = 0$.  

Throughout the paper, we also use $Z$ to denote a standard normal $\mathcal{N}(0, 1)$ random variable having cumulative distribution function and probability density function $\Phi( \cdot )$ and $\phi ( \cdot)$, respectively.

\section{The Normal-Beta Prime (NBP) Prior} \label{NBPPrior}

Suppose we observe $\bm{X} \sim \mathcal{N} (\bm{\theta}, \bm{I}_n)$, and our task is to perform signal detection on the $n$-dimensional vector, $\bm{\theta}$. Consider putting the normal-beta prime (NBP) prior on each $\theta_i, i = 1, \ldots, n$, as follows:
\begin{equation} \label{NBPhier}
\begin{array}{c}
\theta_i | \sigma_i^2 \sim \mathcal{N}(0, \sigma_i^2), i = 1, \ldots, n, \\
\sigma_i^2 \sim \beta'(a, b), i = 1, \ldots, n,
\end{array} 
\end{equation} 
where $\beta'(a,b)$ denotes the beta prime density,
\begin{equation} \label{TPBN}
\pi(\sigma_i^2) = \frac{\Gamma(a+b)}{\Gamma(a) \Gamma(b)} (\sigma_i^2)^{a-1} (1 + \sigma_i^2)^{-(a+b)}, i = 1, ..., n,
\end{equation}
and $a > 0, b>0$. We point out that \cite{ArmaganClydeDunson2011} also considered the beta prime prior as a prior in a normal scale-mixture model. Specifically,  \cite{ArmaganClydeDunson2011} proposed the prior, $\theta_i \sim \mathcal{N}(0,  \lambda_i \tau$) with (\ref{TPBN}) as the prior for the local scale parameters, $\lambda_i \sim \pi(\lambda_i)$, and an additional variance rescaling parameter $\tau > 0$. They called their model the three parameter beta normal (TPBN) prior. Thus, the NBP model can be thought of as a special case of the TPBN prior with $\tau = 1$.   Our work differs from \cite{ArmaganClydeDunson2011} in that \cite{ArmaganClydeDunson2011} recommended fixing the hyperparameters $(a,b)$ \textit{a priori} and controlling the sparsity of the model through the variance rescaling parameter $\tau$. In contrast, we recommend fixing $\tau = 1$ and controlling the sparsity in our model through the hyperparameters $(a, b)$.

Under the NBP model, the priors are \textit{a priori} independent, so the posterior mean of $\theta_i$ under (\ref{NBPhier}) is given by
\begin{equation} \label{postmeanNBP}
\mathbb{E} \{ \mathbb{E} ( \theta_i | X_i, \sigma_i^2 ) \} = \left\{ \mathbb{E} ( 1 - \kappa_i ) | X_i \right\} X_i, 
\end{equation}
where $\kappa_i = 1/(1 + \sigma_i^2)$. Using a simple transformation of variables, we also see that the posterior density of the shrinkage factor $\kappa_i$ is proportional to
\begin{equation} \label{kappadensity}
\pi(\kappa_i | X_i) \propto \exp \left( - \frac{\kappa_i X_i^2}{2} \right) \kappa_i^{b - 1/2} (1 - \kappa_i)^{a - 1}, \hspace{.3cm} \kappa_i \in (0, 1).
\end{equation} 
From (\ref{postmeanNBP}) and (\ref{kappadensity}), it is clear that the amount of posterior shrinkage is controlled by the shrinkage factor $\kappa_i$. For example, with $a = b = 0.5$, we obtain the standard half-Cauchy density $\mathcal{C}^+ (0,1)$ for $\sigma_i$. As noted by \cite{CarvalhoPolsonScott2010} and \cite{PolsonScott2012}, when $\mathcal{C}^+ (0,1)$ is used as the prior for $\sigma_i$ in (\ref{scalemixture}), the marginal density for a single $\theta$ is unbounded at zero. In the next proposition, we show that for \textit{any} choice of $a \in (0, 1/2]$, the marginal distribution for $\theta$ under the NBP prior also has a singularity at zero.

\begin{proposition}
\label{Prop:1}
 Let $\theta$ be an individual unknown population mean in (\ref{X=theta+eps}). If $\theta$ is endowed with the NBP prior (\ref{NBPhier}), then the marginal distribution of $\theta$  is unbounded with a singularity at zero for any $0 < a \leq 1/2$. 
\end{proposition} 

\begin{figure}[t!]
  \centering
  \includegraphics[width=3.8in]{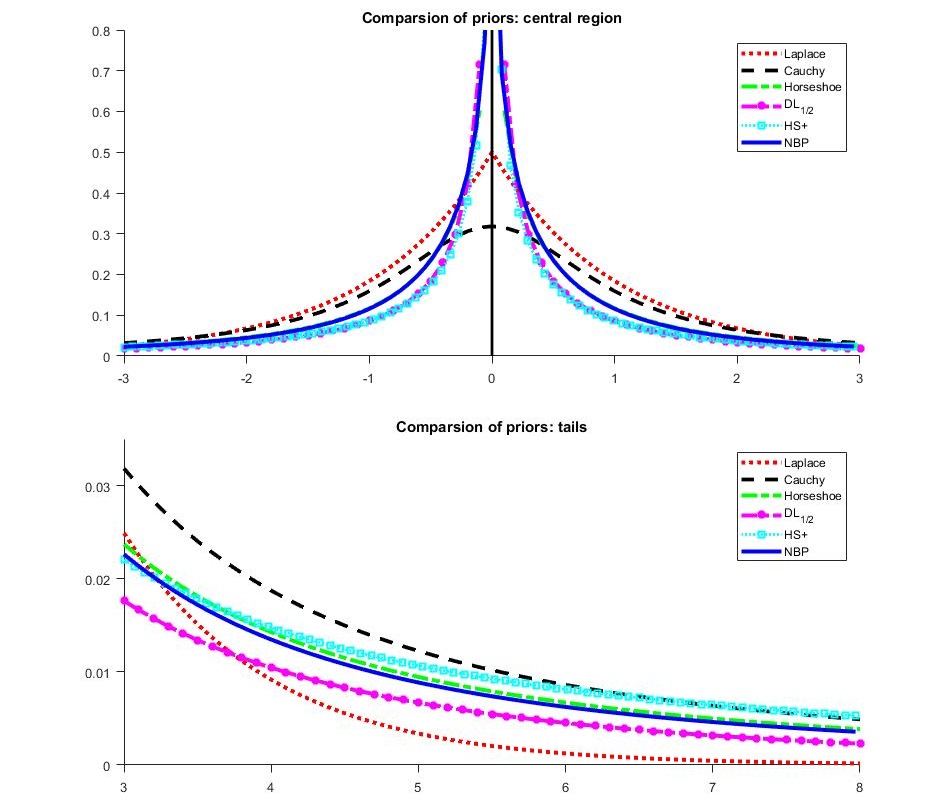}
\caption{Marginal density of the NBP prior (\ref{NBPhier}) with hyperparameters $a = 0.48, b = 0.52$, in comparison to other shrinkage priors. The HS+ prior is the marginal density of the horseshoe+, and the $\textrm{ DL}_{1/2}$ prior is the marginal density for the Dirichlet-Laplace density with $\mathcal{D}(1/2, ...., 1/2)$ specified as a prior in the Bayesian hierarchy.}
\label{fig:1}
\end{figure}

Proposition \ref{Prop:1} gives us some insight into how we should choose the hyperparameters in (\ref{NBPhier}). Namely, we see that for small values of $a$, the NBP prior has a singularity at zero, similar to the horseshoe and the Dirichlet-Laplace \cite{BhattacharyaPatiPillaiDunson2015} priors. Thus, small values of $a$ enable the NBP to obtain sparse estimates of the $\theta_i$'s by shrinking most observations to zero. As we will illustrate in Section \ref{ConcentrationProperties}, the tails of the NBP prior are still heavy enough to identify signals that are significantly far away from zero.

Figure \ref{fig:1} gives a plot of the marginal density $\pi(\theta)$ for the NBP prior (\ref{NBPhier}), with $a = 0.48$ and $b = 0.52$. Figure \ref{fig:1} shows that with small values of $a$ and $b$, the NBP has a singularity at zero, but it maintains the same tail robustness as other well-known shrinkage priors.

\subsection{Concentration Properties of the NBP Prior} \label{ConcentrationProperties}
Consider the NBP prior given in (\ref{NBPhier}), but suppose that we allow the hyperparameter $a \equiv a_n$ to vary with $n$ as $n \rightarrow \infty$. Namely, we allow $0 <a_n <1$ for all $n$, but $a_n \rightarrow 0$ as $n \rightarrow \infty$ so that even more mass is placed around zero as $n \rightarrow \infty$. We also fix $b$ to lie in the interval $(1/2, \infty)$. To emphasize that the hyperparameter $a_n$ depends on $n$, we rewrite the prior (\ref{NBPhier}) as
\begin{equation} \label{NBPn}
\begin{array}{c}
\theta_i | \sigma_i^2 \sim \mathcal{N}(0, \sigma_i^2), i = 1, \ldots, n, \\
\sigma_i^2 \sim \beta' (a_n, b), i = 1, \ldots, n, \\
\end{array}
\end{equation} 
where $a_n \in ( 0, 1)$ with $a_n = o(1)$ and $b \in (1/2, \infty)$. For the rest of the paper, we label this particular variant of the NBP prior as the $\textrm{NBP}_{n}$ prior. 

As described in Section \ref{NBPPrior}, the shrinkage factor $\kappa_i = 1/(1+ \sigma_i^2)$ plays a critical role in the amount of shrinkage of each observation $X_i$.  In this section, we further characterize the tail properties of the posterior distribution $\pi (\kappa_i | X_i)$. Our theoretical results demonstrate that the $\textrm{NBP}_{n}$ prior (\ref{NBPn}) shrinks most estimates of $\theta_i$'s to zero but still has heavy enough tails to identify true signals. In the following results, we assume the $\textrm{NBP}_{n}$ prior on $\theta_i$ for $X_i \sim \mathcal{N}(\theta_i, 1)$.

\begin{theorem} 
\label{Th:1}
For any $a_n, b \in (0, \infty)$,
\begin{equation*}
\mathbb{E} (1 - \kappa_i | X_i) \leq e^{X_i^2/2} \left( \frac{a_n}{a_n + b + 1/2} \right).
\end{equation*}
\end{theorem}

\begin{corollary}
\label{Corollary:1}
If $a_n \rightarrow 0$  as $n \rightarrow \infty$ and $b>0$ is fixed, then $\mathbb{E}(1 - \kappa_i | X_i ) \rightarrow 0$ as $n \rightarrow \infty$.
\end{corollary}

\begin{theorem} 
\label{Th:2}
Fix $\epsilon \in (0, 1)$. For any $a_n \in (0, 1)$, $b \in (1/2, \infty)$,
\begin{equation*}
\Pr (\kappa_i < \epsilon | X_i ) \leq  e^{X_i^2/2} \frac{    a_n \epsilon }{ \left( b + 1/2 \right) (1 - \epsilon)}.
\end{equation*}
\end{theorem}

\begin{corollary} 
\label{Corollary:2}
If $a_n \rightarrow 0$ as $n \rightarrow \infty$ and $b \in ( 1/2, \infty)$ is fixed, then by Theorem \ref{Th:2}, $\Pr(\kappa_i \geq \epsilon | X_i) \rightarrow 1$ for any fixed $\epsilon \in (0, 1)$.
\end{corollary}

\begin{theorem} 
\label{Th:3}
Fix $\eta \in (0,1), \delta \in (0,1)$. Then for any $a_n \in (0,1)$ and $b \in (1/2, \infty)$,
\begin{equation*}
\Pr( \kappa_i > \eta | X_i) \leq  \frac{ \left( b + \frac{1}{2} \right) (1 - \eta)^{a_n} }{ a_n ( \eta \delta)^{b + \frac{1}{2}} } \exp \left( - \frac{\eta (1 - \delta)}{2} X_i^2 \right).
\end{equation*}
\end{theorem}

\begin{corollary}
\label{Corollary:3} 
For any fixed $n$ where $a_n \in (0, 1), b \in (1/2, \infty)$, and for every fixed $\eta \in (0,1)$, $\Pr( \kappa_i \leq \eta | X_i ) \rightarrow 1$ as $X_i \rightarrow \infty$.
\end{corollary}
\begin{corollary}
\label{Corollary:4}
For any fixed $n$ where $a_n \in (0, 1), b \in (1/2, \infty)$, and for every fixed $\eta \in (0,1)$, $\mathbb{E}(1 - \kappa_i | X_i ) \rightarrow 1$ as $X_i \rightarrow \infty$.
\end{corollary}

\noindent Since $\mathbb{E} ( \theta_i | X_i) = \left\{ \mathbb{E} (1 - \kappa_i)  | X_i \right\} X_i$, Corollaries \ref{Corollary:1} and \ref{Corollary:2} illustrate that all observations will be shrunk towards the origin under the $\textrm{NBP}_n$ prior (\ref{NBPn}). However, Corollaries \ref{Corollary:3} and \ref{Corollary:4} demonstrate that if $X_i$ is big enough, then the posterior mean $ \left\{ \mathbb{E} (1 - \kappa_i)  | X_i \right\} X_i \approx X_i$. This ensures the tails of the NBP prior are still sufficiently heavy to detect true signals.

A referee has pointed out that the conditions on the hyperparameter $a_n$ in Corollaries \ref{Corollary:1}-\ref{Corollary:4} closely resemble conditions on the rescaling (or the `global') parameter $\tau \equiv \tau_n$ in priors of the form (\ref{globallocal}) in the literature. Indeed, if $\tau_n \in (0,1)$ and $\tau_n \rightarrow 0$ in (\ref{globallocal}), then one obtains analogous results for priors of the form (\ref{globallocal}). See, e.g. \cite{DattaGhosh2013, GhoshTangGhoshChakrabarti2016}. This is because, as seen in Proposition \ref{Prop:1}, the hyperparameter $a_n$ controls the amount of mass around zero for the NBP (with smaller values leading to heavier mass in the neighborhood of zero). At the same time, keeping $b$ fixed in the range $(1/2, \infty)$ ensures that the NBP has heavy enough tails to prevent overshrinkage of large signals. Thus, the hyperparameter $b$ also plays a similar role as the `local' parameter $\lambda_i$ in (\ref{globallocal}).

\section{Multiple Testing with the NBP Prior} \label{NBPTesting}
\subsection{Asymptotic Bayes Optimality Under Sparsity} \label{ABOS}
Suppose we observe $\bm{X} = (X_1, \ldots, X_n)$, such that $X_i \sim \mathcal{N}(\theta_i, 1)$, for $i = 1, \ldots, n.$ To identify the true signals in $\bm{X}$, we conduct $n$ simultaneous tests: $H_{0i}: \theta_i = 0$ against $H_{1i}: \theta_i \neq 0$, for $i = 1, \ldots, n$. For each $i$, $\theta_i$ is assumed to come from the model,
\begin{equation} \label{thetaspikeslab}
\theta_i \overset{i.i.d.}{\sim} (1-p) \delta_{ \{ 0 \} } + p \mathcal{N}(0, \psi^2), i = 1, \ldots, n,
\end{equation}
where $\psi^2 > 0$ represents a diffuse `slab' density. This point mass mixture model is often considered to be a data generating mechanism for sparse vectors $\bm{\theta}$ in the statistical literature. \cite{CarvalhoPolsonScott2009} referred to model (\ref{thetaspikeslab}) as a `gold standard' for sparse problems.

Model (\ref{thetaspikeslab}) is equivalent to assuming that for each $i$, $\theta_i$ follows a random variable whose distribution is determined by the latent binary random variable $\nu_i$, where $\nu_i = 0$ denotes the event that $H_{0i}$ is true, while $\nu_i =1 $ corresponds to the event that $H_{0i}$ is false. Here $\nu_i$'s are assumed to be i.i.d. Bernoulli($p$) random variables, for some $p$ in $(0,1)$. Under $H_{0i}$, i.e. $\theta_i \sim \delta_{ \{ 0 \} }$, the distribution having a mass 1 at 0, while under $H_{1i}$, $\theta_i \neq 0$ and is assumed to follow an $\mathcal{N}(0, \psi^2)$ distribution with $\psi^2 > 0$. The marginal distributions of the $X_i$'s are then given by the following two-groups model:
\begin{equation} \label{marginalsspikeandslab}
X_i \overset{i.i.d.}{\sim} (1-p) \mathcal{N}(0, 1) + p \mathcal{N}(0, 1+ \psi^2), i = 1, \ldots, n.
\end{equation}
Our testing problem is now equivalent to testing simultaneously
\begin{equation} \label{simultaneoustests}
H_{0i}: \nu_i = 0 \textrm{ versus } H_{1i}: \nu_i = 1 \textrm{ for } i = 1, \ldots, n.
\end{equation}
We consider a symmetric 0-1 loss for each individual test. The total loss of a multiple testing procedure is assumed to be the sum of the individual losses incurred in each test. Letting $t_{1i}$ and $t_{2i}$ denote the probabilities of type I and type II errors of the $i$th test respectively, the Bayes risk of a multiple testing procedure under the two-groups model (\ref{marginalsspikeandslab}) is then given by
\begin{equation} \label{twogroupsrisk}
R = \displaystyle\sum_{i=1}^{n} \{ (1-p) t_{1i} + p t_{2i}  \}.
\end{equation}
\cite{BogdanChakrabartiFrommletGhosh2011} showed that the rule which minimizes the Bayes risk in (\ref{twogroupsrisk}) is the test which, for each $i = 1, \ldots, n$, rejects $H_{0i}$ if 
\begin{equation} \label{BayesOracle}
\frac{ f(x_i | \nu_i = 1)}{ f(x_i | \nu_i = 0) } > \frac{1-p}{p}, \textrm{i.e. } X_i^2 > c^2, 
\end{equation}
where $f(x_i | \nu_i = 1)$ denotes the marginal density of $X_i$ under $H_{1i}$, while $f(x_i | \nu_i = 0)$ denotes that under $H_{0i}$ and 
\begin{equation*}
c^2 \equiv c_{\psi, f}^2 = \frac{1+ \psi^2}{\psi^2} ( \log(1+\psi^2) + 2 \log (f)),
\end{equation*}
with $f = (1- p)/p$. The above rule is known as the Bayes Oracle, because it makes use of unknown parameters $\psi$ and $p$. By reparametrizing as $u = \psi^2$ and $v = u f^2$, the above threshold becomes
\begin{equation*}
c^2 \equiv c_{u,v}^2 = \left( 1 + \frac{1}{u} \right) \left( \log v + \log \left( 1 + \frac{1}{u} \right) \right).
\end{equation*}
\cite{BogdanChakrabartiFrommletGhosh2011} considered the following asymptotic scheme.
\vspace{-.5cm}
\subsection*{Assumption 1} 
The sequences of vectors $( \psi_n, p_n )$ satisfies the following conditions:
\begin{enumerate} 
\item
$p_n \rightarrow 0$ as $n \rightarrow \infty$. \label{Assumption:1}
\item
$u_n = \psi_n^2 \rightarrow \infty$ as $n \rightarrow \infty$.
\item
$v_n = u_n f^2 \rightarrow \infty$ as $n \rightarrow \infty$.
\item
$\frac{ \log v_n}{u_n} \rightarrow C \in (0, \infty)$ as $n \rightarrow \infty$.
\end{enumerate}
The first condition in Assumption 1 assumes that the underlying $\bm{\theta}$ becomes more sparse as $n$ approaches infinity, while the second condition ensures that true signals can still be identified. \noindent \cite{BogdanChakrabartiFrommletGhosh2011} provided detailed insight  on the threshold $C$ arising from the third and fourth conditions. Summarizing briefly, if $C = 0$, then the probability of a Type I error is one and the probability of a Type II error is zero. If $C = \infty$, then the probability of a Type I error is zero and the probability of a Type II error is one. Under Assumption \ref{Assumption:1},  \cite{BogdanChakrabartiFrommletGhosh2011} showed that the corresponding asymptotic Bayes Oracle risk has a particularly simple form, which is given by
\begin{equation} \label{ABOSrisk}
R_{Opt}^{BO} = n((1-p) t_1^{BO} + p t_2^{BO} ) = np(2 \Phi (\sqrt{C}) - 1)(1+ o(1)),
\end{equation}
where the $o(1)$ terms tend to zero as $n \rightarrow \infty$. A testing procedure with risk $R$ is said to be asymptotically Bayes optimal under sparsity (ABOS) if 
\begin{equation} \label{ABOSmult}
\frac{R}{R_{Opt}^{BO}} \rightarrow 1 \textrm{ as } n \rightarrow \infty.
\end{equation}
\begin{remark}
	\cite{BogdanChakrabartiFrommletGhosh2011} consider the more general case where under the null hypothesis, $H_{0i}: \theta_i \sim \mathcal{N}(0, \zeta^2)$ with $0 \leq \zeta \ll \psi$. That is, \cite{BogdanChakrabartiFrommletGhosh2011} assumed that the true data-generating mechanism for $\bm{\theta}$ is given by
	\begin{equation} \label{thetaspikeslab2}
\theta_i \overset{i.i.d.}{\sim} (1-p) \mathcal{N}(0, \zeta^2) + p \mathcal{N} (0, \psi^2), i = 1, \ldots, n.
\end{equation}
	The point mass mixture (\ref{thetaspikeslab}) is obtained as a special case of (\ref{thetaspikeslab2}) by setting $\zeta = 0$. \cite{BogdanChakrabartiFrommletGhosh2011} showed that the asymptotic Bayes Oracle risk under (\ref{thetaspikeslab2}) is the same as (\ref{ABOSrisk}) if we replace $u = \psi^2$ in Assumption \ref{Assumption:1} with $u = (\psi / (\zeta + 1) )^2$. If we assume that the true $\bm{\theta}$ comes from (\ref{thetaspikeslab2}) with $\zeta > 0$ and we similarly set $u = (\psi / (\zeta + 1))^2$ in Assumption \ref{Assumption:1}, then all the results in this manuscript will continue to hold. Thus, when $\bm{\theta}$ is `nearly' (but not exactly) sparse, thresholding rule (\ref{thresholdingrule}) for classifying signals under the NBP prior is also ABOS. 
\end{remark}

In Sections \ref{NBPTestingNonAdaptive} and \ref{NBPTestingAdaptive}, we consider two thresholding rules based on the NBP model. In the first case, we assume the sparsity level $p$ under the true data-generating model (\ref{thetaspikeslab}) to be known. For the more realistic scenario where $p$ is unknown, we base our test procedure on a data-driven estimate of $p$. Since we estimate the unknown proportion of signals from the data, we refer to this latter procedure as a \textit{data-adaptive} testing rule. 

\subsection{Related Work for Scale-Mixture Shrinkage Priors} \label{RelatedWork}

We briefly survey related work on multiple testing under normal scale-mixture shrinkage priors (\ref{scalemixture}) to demonstrate the novelty of our results. \cite{GhoshChakrabarti2017} showed that for priors of the form (\ref{globallocal}), thresholding rule (\ref{thresholdingrule}) is ABOS provided that: a) the variance rescaling parameter $\tau > 0$ decays at an appropriate rate or is estimated by an appropriate plug-in estimator, and b) the slowly varying component $L(\cdot)$ in the scale prior, $\pi(\lambda_i) \propto \lambda_i^{-a-1} L(\lambda_i)$, is uniformly bounded above and below on the interval $\lambda_i \in (0, \infty)$. The NBP prior (\ref{NBPhier}) does not require a rescaling parameter $\tau > 0$ in the normal variance in the first level of the Bayes hierarchy. Thus, our results cannot be obtained from those in \cite{GhoshChakrabarti2017}.

Under certain conditions on the prior for the scale parameter $\sigma_i^2$ in (\ref{scalemixture}), \cite{Salomond2017} derived asymptotic upper bounds on Type I and Type II errors and the Bayes risk (\ref{twogroupsrisk}) for both non-adaptive and data-adaptive test procedures induced by scale-mixture shrinkage priors. Under these conditions, the upper bound on the Bayes risk for scale-mixture priors is of the same order as the Bayes Oracle risk. Specifically, \cite{Salomond2017} showed that the Bayes risk (\ref{twogroupsrisk}) for thresholding rule (\ref{genericthreshold})  can be bounded from above by $np[16 \sqrt{\pi} C / c + 2 \Phi ( \sqrt{2 K (u_0 +1 ) C}) - 1](1+o(1))$ for known $p$ and by $np[ 16 \sqrt{\pi} C D/c + 2 \Phi (\sqrt{2 K (u_0 + 1)(1+\xi)C }) - 1 ] (1+o(1))$ for unknown $p$, where $C$ is the constant from the fourth condition in Assumption 1 and $c > 0, K \geq 0, u_0 > 0, D>0, \xi \geq 0$ are appropriate constants that depend on the prior.

One can show that with appropriate conditions on the hyperparameters $(a, b)$, the NBP prior (\ref{NBPhier}) satisfies the conditions in \cite{Salomond2017}. Therefore, our prior can also obtain the upper bound on the risk derived by \cite{Salomond2017}. However, the results that we present in this paper do not immediately follow from \cite{Salomond2017} because: a) we provide \textit{lower} bounds on Type I and Type II errors under our prior, and b) we establish that the Bayes risk under the NBP prior is actually asymptotically the same as that of the Bayes Oracle risk given in (\ref{ABOSrisk}). Therefore, our bounds are provably sharper than those of \cite{Salomond2017}. 

\subsection{A Non-Adaptive Testing Rule Under the NBP Prior} \label{NBPTestingNonAdaptive}
As noted earlier, the posterior mean under the NBP prior depends heavily on the shrinkage factor, $\kappa_i = 1/ (1+\sigma_i^2)$. Because of the concentration properties of the NBP prior proven in Section \ref{ConcentrationProperties}, a sensible thresholding rule classifies observations as signals or as noise based on the posterior distribution of this shrinkage factor.  Consider the following testing rule for the $i$th observation $X_i$:
\begin{equation} \label{thresholdingrule}
\textrm{Reject } H_{0i} \textrm{ if } \mathbb{E} (1 - \kappa_i | X_i ) > \frac{1}{2},
\end{equation}
where $\kappa_i$ is the shrinkage factor based on the $\textrm{NBP}_n$ prior (\ref{NBPn}). Within the context of multiple testing, a good benchmark for our test procedure (\ref{thresholdingrule}) should be whether it is ABOS, i.e. whether its optimal risk is asymptotically equal to that of the Bayes Oracle risk (\ref{ABOSrisk}). Adopting the framework of \cite{BogdanChakrabartiFrommletGhosh2011}, we let $R_{NBP}$ denote the asymptotic Bayes risk of testing rule (\ref{thresholdingrule}). In the next four theorems, we derive sharp lower and upper bounds on the Type I and Type II error probabilities for test procedure (\ref{thresholdingrule}). These error probabilities are given by
\begin{equation} \label{errorprobabilities}
\begin{array}{rcl}
t_{1i} & = & \Pr \left[ \mathbb{E}(1 - \kappa_i | X_i) > \frac{1}{2} \bigg| H_{0i} \textrm{ is true} \right], \\
&&\\
t_{2i} & = & \Pr \left[ \mathbb{E}(1 - \kappa_i | X_i) \leq \frac{1}{2} \bigg| H_{1i} \textrm{ is true} \right].
\end{array}
\end{equation}

\begin{theorem}
\label{Theorem:4}
Suppose that $X_1, \ldots, X_n$ are i.i.d. observations having distribution (\ref{marginalsspikeandslab}) where the sequence of vectors $(\psi_n^2, p_n)$ satisfies Assumption \ref{Assumption:1}. Suppose we wish to test (\ref{simultaneoustests}) using the classification rule (\ref{thresholdingrule}) under the $\textrm{NBP}_n$ prior. Then for all $n$, an upper bound for the probability of a Type I error for the $i$th test is given by
\begin{equation*}
t_{1i} \leq \frac{2 \sqrt{2} a_n}{\sqrt{\pi} (a_n + b + 1/2)} \left[ \log \left( \frac{a_n + b + 1/2}{2 a_n } \right) \right]^{-1/2}.
\end{equation*} 
\end{theorem}

\begin{theorem}
\label{Theorem:5}
Assume the same setup of Theorem \ref{Theorem:4}. Suppose we wish to test (\ref{simultaneoustests}) using the classification rule (\ref{thresholdingrule}) under the $\textrm{NBP}_n$ prior. Suppose further that $a_n \in (0, 1)$, with $a_n \rightarrow 0$ as $n \rightarrow \infty$ and  $b \in (1/2, \infty)$ is fixed. Then for any $\xi \in (0, 1/2)$ and $\delta \in (0, 1)$, a lower bound for the probability of a Type I error for the $i$th test as $n \rightarrow \infty$ is given by
\begin{equation*}
t_{1i} \geq 1 - \Phi \left( \sqrt{ \frac{2}{\xi ( 1 -\delta)} \left[ \log \left( \frac{ \left( b + \frac{1}{2} \right) (1 - \xi)^{a_n} }{ a_n ( \xi \delta)^{b + \frac{1}{2}} } \right) \right] } \right).
\end{equation*}
\end{theorem}

\begin{theorem}
\label{Theorem:6}
Assume the same setup as Theorem \ref{Theorem:4}. Suppose we wish to test (\ref{simultaneoustests}) using the classification rule (\ref{thresholdingrule}) under the $\textrm{NBP}_n$ prior. Suppose further that $a_n \in (0, 1)$, with $a_n \rightarrow 0$ as $n \rightarrow \infty$ in such a way that $\lim_{n \rightarrow \infty} a_n / p_n \in (0, \infty)$ and that $b \in (1/2, \infty)$ is fixed. Fix $\eta \in (0, 1), \delta \in (0, 1)$, and choose any $\rho > 2/(\eta (1 - \delta))$. Then as $n \rightarrow \infty$, an upper bound for the probability of a Type II error for the $i$th test is given by
\begin{equation*}
t_{2i} \leq \left[ 2 \Phi \left(\sqrt{ \frac{\rho C}{2}}\right) - 1 \right] (1+o(1)) \textrm{ as } n \rightarrow \infty,
\end{equation*}
where the $o(1)$ terms above go to 0 as $n \rightarrow \infty$.
\end{theorem}

\begin{theorem}
\label{Theorem:7}
Assume the same setup as Theorem \ref{Theorem:4}. Suppose we wish to test (\ref{simultaneoustests}) using the classification rule (\ref{thresholdingrule}) under the $\textrm{NBP}_n$ prior. Suppose further that $a_n \in (0, 1)$, with $a_n \rightarrow 0$ as $n \rightarrow \infty$ in such as way that $\lim_{n \rightarrow \infty} a_n / p_n \in (0, \infty)$ and that $b \in (1/2, \infty)$ is fixed. Then as $n \rightarrow \infty$, a lower bound for the probability of a Type II error for the $i$th test is given by
\begin{equation*}
t_{2i} \geq \left[ 2 \Phi (\sqrt{C}) - 1 \right] (1+ o(1)) \textrm{ as } n \rightarrow \infty,
\end{equation*}
where the $o(1)$ terms tend to zero as $n \rightarrow \infty$.
\end{theorem}

Theorems \ref{Theorem:4}-\ref{Theorem:5} show that for \textit{any} sequence of hyperparameters $a_n$ such that $a_n \rightarrow 0$ as $n \rightarrow \infty$, the probability of a Type I error for test (\ref{thresholdingrule}) is asymptotically vanishing under the $\textrm{NBP}_n$ prior.  Meanwhile, Theorems \ref{Theorem:6}-\ref{Theorem:7} show that if $a_n$ is the same order as the true signal proportion $p_n$, then the probability of Type II error for test (\ref{thresholdingrule}) can be bounded from above and below. Notice that in Theorem \ref{Theorem:6}, we are free to choose any $\rho$ arbitrarily close to 2 in the upper bound on the probability of a Type II error, with 2 being the infimum for $\rho$. Thus, the limit inferior of upper bound in Theorem \ref{Theorem:6} is the same as the lower bound established in Theorem \ref{Theorem:7}, and so these bounds are sharp. Altogether, Theorems \ref{Theorem:4}-\ref{Theorem:7} show that asymptotically, the Bayes risk (\ref{ABOSrisk}) is controlled entirely by the Type II error. If $C \approx 0$ in Assumption \ref{Assumption:1}, then the power of the $i$th test, $1-t_{{2i}}$, under the $\textrm{NBP}_n$ prior (\ref{NBPn}) will be close to one.

Having obtained appropriate upper and lower bounds on the Type I and Type II probabilities under thresholding rule (\ref{thresholdingrule}), we are ready to state our main theorem which proves that our method under the $\textrm{NBP}_n$ prior is asymptotically Bayes optimal under sparsity.
\begin{theorem}
\label{Th:8}
Suppose that $X_1, \ldots, X_n$ are i.i.d. observations having distribution (\ref{marginalsspikeandslab}) where the sequence of vectors $(\psi^2, p)$ satisfies Assumption \ref{Assumption:1}. Suppose we wish to test (\ref{simultaneoustests}) using the classification rule (\ref{thresholdingrule}). Suppose further that $a_n \in (0, 1)$, with $a_n \rightarrow 0$ as $n \rightarrow \infty$ in such a way that $\lim_{n \rightarrow \infty} a_n / p_n \in (0, \infty)$ and that $b \in ( 1/2, \infty)$ is fixed. Then
\begin{equation} \label{NBPriskratio}
\lim_{n \rightarrow \infty} \frac{R_{NBP}}{R_{Opt}^{BO}} = 1,
\end{equation}
i.e. rule (\ref{thresholdingrule}) based on the $\textrm{NBP}_n$ prior (\ref{NBPn}) is ABOS.
\end{theorem}

We have shown that our thresholding rule based on the $\textrm{NBP}_n$ prior asymptotically attains the Bayes Oracle risk exactly, provided that $a_n$ is of the same order as the sparsity level $p_n$. Since $p_n$ is typically unknown, it should ideally be estimated from the data, and our theoretical findings suggest how to build adaptive procedures for setting $a_n$, which we describe in Sections \ref{NBPTestingAdaptive} and \ref{DataAdaptiveMethods}.

\subsection{A Data-Adaptive Testing Rule Under the NBP Prior} \label{NBPTestingAdaptive}

As we found in Theorem \ref{Th:8}, our test procedure (\ref{thresholdingrule}) has the Bayes oracle property under the $\textrm{NBP}_n$ prior, provided that $a_n$ is of the same order as the true signal proportion $p_n$ and $b \in (1/2, \infty)$ is fixed. However, $p_n$ is typically unknown, and as a result, we must estimate it from the data. To this end, we use the estimator proposed by \cite{VanDerPasKleijnVanDerVaart2014}:
\begin{equation} \label{EBsparsity}
\widehat{a}_n^{ES} := \max \left\{ \frac{1}{n}, \frac{1}{c_2 n} \displaystyle \sum_{j=1}^n 1 \{ |X_j| > \sqrt{c_1 \log n} \} \right\},
\end{equation}
where $c_1 \geq 2$ and $c_2 \geq 1$ are fixed constants, and we use $ES$ to denote `estimated sparsity.' This choice is motivated by the so-called `universal threshold,' $\sqrt{2 \log n}$. It is well-known that signals which fall below this threshold are shrunk towards zero, and thus, they may not be detected. 

Based on the considerations above, we now introduce a data-adaptive testing rule under the $\textrm{NBP}_n$ prior. Setting $a_n \equiv \widehat{a}_n^{ES}$ and $b \in (1/2, \infty)$ as the hyperparameters in the NBP prior, our test for the $i$th observation $X_i$ is:
\begin{equation} \label{thresholdingruleEB}
\textrm{Reject } H_{0i} \textrm{ if } \mathbb{E} (1 - \kappa_i | X_i , \widehat{a}_n^{ES}) > \frac{1}{2},
\end{equation}
From a decision theoretic perspective, we now demonstrate that setting the hyperparameter $a_n$ equal to $\widehat{a}_n^{ES}$ is also justified. Letting $R_{NBP}^{ES}$ denote the asymptotic Bayes risk, we first derive sharp lower bounds and upper bounds on the Type I and Type II error probabilities, which we denote as $\widetilde{t}_{1i}$ and $\widetilde{t}_{2i}$ respectively. We then illustrate that testing rule (\ref{thresholdingruleEB}) is \textit{also} ABOS.

Following the notation of \cite{GhoshTangGhoshChakrabarti2016}, we denote
\begin{equation} \label{alphan}
\alpha_n = \Pr(|X_i| > \sqrt{c_1 \log n}), \textrm{ and } \beta = 1 - \Phi(c_1 C/2 \epsilon),
\end{equation}
where $\epsilon \in (0, 1)$, $c_1$ is the constant in (\ref{EBsparsity}), and $C$ is the constant from Assumption \ref{Assumption:1}. In \cite{GhoshTangGhoshChakrabarti2016}, it was shown that as long as the signal proportion $p_n \propto n^{-\epsilon}$ and Assumption \ref{Assumption:1} holds, then
\begin{equation} \label{alpharelation}
\alpha_n = 2 \beta p_n (1 + o(1)),
\end{equation}
under the two-groups model (\ref{marginalsspikeandslab}), where the $o(1)$ terms go to 0 as $n \rightarrow \infty$. We will use (\ref{alphan}) and (\ref{alpharelation}) to prove Theorems \ref{Theorem:9} and \ref{Theorem:10}, which provide asymptotic bounds on the Type I and Type II error probabilities under (\ref{thresholdingruleEB}).

\begin{theorem}
\label{Theorem:9}
Suppose that $X_1, \ldots, X_n$ are i.i.d. observations having distribution (\ref{marginalsspikeandslab}) where the sequence of vectors $(\psi_n^2, p_n)$ satisfies Assumption \ref{Assumption:1}, with $p_n \propto n^{-\epsilon}, \epsilon \in (0, 1)$. Fix $b \in (1/2, \infty)$, $c_1 \geq 2$, $c_2 \geq 1$, $\xi \in (0, 1/2)$, and $\delta \in (0, 1)$. Suppose we wish to test (\ref{simultaneoustests}) using the classification rule (\ref{thresholdingruleEB}).  Then as $n \rightarrow \infty$, bounds for the probability of a Type I error for the $i$th test, $\widetilde{t}_{1i}$, are given by

\begin{align*}
& 1 - \Phi \left( \sqrt{ \frac{2}{\xi ( 1 -\delta)} \left[ \log \left( \frac{ \left( b + \frac{1}{2} \right) (1 - \xi)^{2 \alpha_n} }{ 2 \alpha_n ( \xi \delta)^{b + \frac{1}{2}} } \right) \right] } \right) \leq \widetilde{t}_{1i} \\
& \qquad \leq \frac{4 \alpha_n}{\sqrt{\pi} ( 2 \alpha_n + b + 1/2)} \left[ \log \left( \frac{2 \alpha_n + b + 1/2}{4 \alpha_n} \right) \right]^{-1/2} (1 + o(1))  \\
& \qquad \qquad + \frac{1 / \sqrt{\pi}}{n^{c_1/2} \sqrt{\log n}} + e^{-2 (2 \log 2 - 1) \beta np_n(1+o(1))},
\end{align*} 
where $\alpha_n$ and $\beta$ are as in (\ref{alphan}).
\end{theorem}

\begin{theorem}
\label{Theorem:10}
Assume the same setup as Theorem \ref{Theorem:9}, and assume that $p_n \propto n^{-\epsilon}, \epsilon \in (0, 1)$.  Fix $b \in (1/2, \infty)$, $c_1 \geq 2$, $c_2 \geq 1$, $\eta \in (0, 1)$, and $\delta \in (0, 1)$, and choose any $\rho > 2/(\eta (1 - \delta))$. Suppose we wish to test (\ref{simultaneoustests}) using the classification rule (\ref{thresholdingruleEB}).  Then as $n \rightarrow \infty$, bounds for the probability of a Type II error for the $i$th test, $\widetilde{t}_{1i}$, are given by
\begin{equation*}
\left[ 2 \Phi (\sqrt{C}) - 1 \right] (1+ o(1)) \leq \widetilde{t}_{2i} \leq  \left[ 2 \Phi \left( \sqrt{\frac{\rho C}{2}} \right) - 1 \right](1+o(1)) \textrm{ as } n \rightarrow \infty,
\end{equation*}
where the $o(1)$ terms tend to zero as $n \rightarrow \infty$.
\end{theorem}

We pause briefly to compare Theorems \ref{Theorem:9}-\ref{Theorem:10} with Theorems \ref{Theorem:4}-\ref{Theorem:7}. Theorems \ref{Theorem:4}-\ref{Theorem:5} demonstrated that for \textit{any} sequence of hyperparameters $a_n$ such that $a_n \rightarrow 0$ as $n \rightarrow \infty$, the probability of a Type I error under thresholding rule (\ref{thresholdingrule}) asymptotically vanishes. Theorem \ref{Theorem:9} shows that this will also be the case for plug-in estimator $\widehat{a}_n^{EB}$ as long as $\alpha_n := \Pr ( |X_i| > \sqrt{c_1 \log n} )$, goes to 0 as $n \rightarrow \infty$. This condition holds for any $p_n \propto n^{-\epsilon}, \epsilon \in (0, 1)$. In replacing the generic sequence $a_n$ with a specific plug-in value $\widehat{a}_n^{EB}$, the bounds in Theorem \ref{Theorem:9} differ in constants from the bounds derived in Theorem \ref{Theorem:4}. However, the bounds in Theorems \ref{Theorem:4} and \ref{Theorem:9} are ultimately of the same order if $a_n \rightarrow 0$ and $\alpha_n \rightarrow 0$. In addition, Theorem \ref{Theorem:10} shows that if $p_n \propto n^{-\epsilon}$ and we utilize the EB estimator $\widehat{a}_n^{ES}$ (\ref{EBsparsity}) in place of $a_n$, then the upper and lower bounds on probability of Type II error are the same as those in Theorems \ref{Theorem:6}-\ref{Theorem:7}.

Having obtained appropriate upper and lower bounds on the Type I and Type II probabilities under thresholding rule (\ref{thresholdingruleEB}), we are ready to state our main theorem which proves that our data-adaptive testing procedure (\ref{thresholdingruleEB}) also possesses the Bayes Oracle property in the entire range of sparsity parameters $p \propto n^{-\epsilon}, \epsilon \in (0, 1)$. 
\begin{theorem}
\label{Theorem:11}
Suppose that $X_1, \ldots, X_n$ are i.i.d. observations having distribution (\ref{marginalsspikeandslab}) where the sequence of vectors $(\psi^2, p)$ satisfies Assumption \ref{Assumption:1}. Further assume that $p \propto n^{-\epsilon}, \epsilon \in (0, 1)$. For the NBP prior (\ref{NBPhier}),  fix $b \in (1/2, \infty)$ and set $a = \widehat{a}_n^{ES}$, where $\widehat{a}_n^{ES}$ is as in (\ref{EBsparsity}), with fixed $(c_1, c_2)$ satisfying $c_1 \geq 2$ and $c_2 \geq 1$. Suppose that we wish to test (\ref{simultaneoustests}) using the classification rule (\ref{thresholdingruleEB}). Then
\begin{equation} \label{NBPEBriskratio}
\lim_{n \rightarrow \infty} \frac{R^{ES}_{NBP}}{R_{Opt}^{BO}} = 1,
\end{equation}
i.e. data-adaptive thresholding rule (\ref{thresholdingruleEB}) is ABOS.
\end{theorem}
\begin{proof}
This follows the same reasoning as the proof for Theorem \ref{Th:8}, except we replace the bounds for $t_{1i}$ and $t_{2i}$ with those of $\widetilde{t}_{1i}$ and $\widetilde{t}_{2i}$ from Theorems \ref{Theorem:9} and \ref{Theorem:10}. To prove that the bounds for $\widetilde{t}_{1i}$ in Theorem \ref{Theorem:9} tend to zero, note that $p_n \propto n^{-\epsilon}, \epsilon \in (0,1)$, and therefore, by (\ref{alpharelation}), $\alpha_n \rightarrow 0$ as $n \rightarrow \infty$.
\end{proof}

Note that this condition on $p$ is quite mild. For comparison, \cite{BogdanChakrabartiFrommletGhosh2011} showed that the widely used Benjamini-Hochberg (BH) \cite{BenjaminiHochberg1995} procedure for controlling the false discovery rate (FDR) is ABOS if and only if $p_n \propto n^{-\epsilon}, \epsilon \in (0, 1]$. Unlike the BH procedure, however, the NBP requires an estimate of the unknown sparsity level $p_n$ in order to achieve the Bayes Oracle property and is \textit{not} ABOS if $p_n = n^{-1}$ (in this case, the probability of a Type I error is not asymptotically vanishing). Nevertheless, there are several advantages of the NBP model over BH. The BH procedure cannot be used for estimation or uncertainty quantification of $\bm{\theta}$. In contrast, the NBP model not only admits a testing procedure that is ABOS, but the NBP posterior can also be used to obtain estimates and credible intervals for $\bm{\theta}$. Further, obtaining an estimate of unknown sparsity $p_n$, such as the one in (\ref{EBsparsity}) or the ones described in Section \ref{DataAdaptiveMethods}, is not computationally expensive; this adds only a single preprocessing step or an extra iteration in the Markov chain Monte Carlo (MCMC) algorithm. The assumption that the true sparsity level satisfies $p \propto n^{-\epsilon}, \epsilon \in (0,1)$ (i.e. that there is more than one signal in the data) is also very likely to be satisfied in practice.

\section{Two Other Data-Adaptive Approaches for Estimating the Sparsity Parameter} \label{DataAdaptiveMethods}
As we demonstrated in Section \ref{NBPTestingNonAdaptive} and \ref{NBPTestingAdaptive}, we can construct hypothesis tests based on the NBP prior which have the Bayes Oracle property by fixing $b \in (1/2, \infty)$ and by choosing $a$ to be comparable to the proportion of true signals. By Proposition \ref{Prop:1}, $a$ also controls the amount of mass around zero. Thus, $a$ can be interpreted as the sparsity parameter, and the ideal choice of $a$ should lie in the range $[1/n, 1]$.

In \cite{VanDerPasSzaboVanDerVaart2017}, the variance rescaling parameter $\tau$ in the horseshoe prior is estimated through restricted marginal maximum likelihood (REML) on the interval $[1/n, 1]$ or by placing a prior on $\tau$ with its support truncated to lie in the interval $[1/n, 1]$. The methods in \cite{VanDerPasSzaboVanDerVaart2017} enable the horseshoe to achieve near-minimax posterior contraction. 

\subsection{A Restricted Marginal Maximum Likelihood (REML) Approach} \label{REML}
Inspired by \cite{VanDerPasSzaboVanDerVaart2017}'s work, we first propose a REML approach to estimating $a$. We take our estimate $\widehat{a}_n^{REML}$ to be the marginal maximum likelihood estimate of $a$ restricted on the interval $[1/n, 1]$. That is, for a fixed $b$, we define $\widehat{a}_n^{REML}$ as 
\begin{equation} \label{remlNBP}
\widehat{a}_n^{REML} = \displaystyle \argmax_{a \in [1/n, 1]} \displaystyle \prod_{i=1}^{n} m(X_i),
\end{equation}
where $m(X_i)$ denotes the marginal density for a single observation $X_i$, i.e.,
\begin{equation} \label{NBPmarginal}
m(X_i) =  \displaystyle \int_{-\infty}^{\infty} \displaystyle \int_{0}^{\infty} \phi(X_i- \theta_i) \phi ( \theta_i / \sigma_i ) \pi(\sigma_i^2) d \sigma_i^2 d \theta_i,
\end{equation} 
and $\pi(\sigma_i^2)$ is the prior for beta prime density given in (\ref{TPBN}). A closed form solution to (\ref{remlNBP}) is unavailable, but it can be computed using numerical integration and optimization. 

We now introduce yet another data-adaptive testing rule under the NBP prior. Suppose that we set $(\widehat{a}_n^{REML}, b)$ as our hyperparameters in the NBP prior (\ref{NBPhier}), where $b \in (1/2, \infty)$. Then our test for the $i$th observation $X_i$ is:
\begin{equation} \label{thresholdingruleREML}
\textrm{Reject } H_{0i} \textrm{ if } \mathbb{E} (1 - \kappa_i | X_i , \widehat{a}_n^{REML}) > \frac{1}{2}.
\end{equation}

\subsection{A Hierarchical Bayes Approach} \label{fullyBayes}

Our results also suggest that if we adopt a fully Bayes approach for estimating the sparsity parameter $a$, the prior on $a$ should have its support truncated to $[1/n, 1]$. Suppose that we fix $b \in (1/2, \infty)$. Our hierarchical model is defined as
\begin{equation} \label{NBPhierFB}
\begin{array}{rl}
\theta_i | \sigma_i^2 & \sim \mathcal{N}(0, \sigma_i^2), i = 1, \ldots, n, \\
\sigma_i^2 &  \sim \beta'(a, b), i = 1, \ldots, n, \\
a & \sim  \pi (a),
\end{array} 
\end{equation} 
where the support of $\pi(a)$ is $[1/n, 1]$. Under (\ref{NBPhierFB}), our thresholding rule now becomes
\begin{equation} \label{thresholdingruleFB}
\textrm{Reject } H_{0i} \textrm{ if } \mathbb{E} (1 - \kappa_i | X_1 , \ldots, X_n ) > \frac{1}{2}.
\end{equation}
Note that because we have placed a prior on $a$, the priors for the $\theta_i$'s are no longer \textit{a priori} independent. Thus, the posterior densities of the $\theta_i$'s (and hence the $\kappa_i$'s) also depend on \textit{all} the data. For our simulation studies, we consider both a uniform prior for $a$, i.e. $a \sim \mathcal{U}(1/n, 1)$, and a standard Cauchy prior for $a$ truncated to $[1/n, 1]$, i.e. $\pi(a) = [\textrm{arctan}(1) -\textrm{arctan}(1/n)]^{-1} (1+a)^{-1} \mathbb{I} \{ 1/n < a < 1 \}$.

In Section \ref{Simulations}, we demonstrate that test procedures (\ref{thresholdingruleREML}) and (\ref{thresholdingruleFB}) both mimic the Bayes Oracle performance in simulations. We hope to provide theoretical justification for (\ref{thresholdingruleREML}) and (\ref{thresholdingruleFB}) in the future.  Following the work of \cite{VanDerPasSzaboVanDerVaart2017} for the horseshoe prior, we believe that useful bounds on (\ref{NBPmarginal}) and on the posterior $\pi(a | X_1, \ldots, X_n)$  under (\ref{NBPhierFB}) can be derived to facilitate theoretical analysis of the NBP prior when $a$ is estimated by REML or by a truncated prior.

\section{Simulation Studies} \label{Simulations}

\subsection{Implementation and Selection of the Hyperparameter $b$}
In the case where $a$ is fixed \textit{a priori} or estimated with a plug-in estimator, the NBP model (\ref{NBPhier}) can be implemented straightforwardly using Gibbs sampling. If a prior is placed on the hyperparameter $a$, as in (\ref{NBPhierFB}), then we use Metropolis-Hastings to update $a$. In the Supplementary Materials, we provide the full details on how to sample from models (\ref{NBPhier}) and (\ref{NBPhierFB}). For the hierarchical Bayes approach, we saw that the MCMC chains mixed well and converged very quickly (in less than 100 iterations), even if we initialized the values to be far away from the truth. This is also illustrated in the Supplementary Materials. We provide the implementation of the NBP model and the multiple testing procedures (\ref{thresholdingruleEB}), (\ref{thresholdingruleREML}), and (\ref{thresholdingruleFB}) in a comprehensive \textsf{R} package, \texttt{NormalBetaPrime}.

In order to use the NBP prior (\ref{NBPhier}) for multiple testing, we recommend setting $b$ to lie in the interval $(1/2, 1/2+\delta]$, for some small $\delta > 0$, and estimating $a$ from the data. We could also estimate $b$ from the data, but our theoretical results in Theorems \ref{Th:8} and Theorem \ref{Theorem:11} demonstrate that asymptotically, the specific choice of $b$ plays no role. As pointed out by \cite{PolsonScott2012}, smaller values of $b$ correspond to heavier tails, with values of $b$ close to 1/2 giving Cauchy-like tails. Based on these considerations, we suggest the default choice of $b = 1/2 + 1/n$, so that the theoretical results established earlier hold, while the tails are still quite heavy. 

\subsection{Simulation Study}
We adopt the simulation framework of \cite{DattaGhosh2013} and \cite{GhoshTangGhoshChakrabarti2016} and fix sparsity levels at $p \in \left\{ 0.01, 0.05, 0.10, 0.15, 0.2, 0.25, 0.3, 0.35, 0.4, 0.45, 0.5 \right\}$, for a total of 11 simulation settings. For sample size $n = 500$ and each $p$, we generate data from the two-groups model (\ref{marginalsspikeandslab}), with $\psi = \sqrt{2 \log n} = 3.53$. We fix $b = 1/2+1/n$ and implement the NBP model with each of the following estimates for $a$:
\begin{enumerate}
\item
NBP-ES: the estimated sparsity (ES) estimator $\widehat{a}^{ES}$, as in (\ref{EBsparsity}), with fixed constants $c_1 = 2, c_2 = 1$.
\item
NBP-REML: the REML estimator $\widehat{a}^{REML}$, as in (\ref{remlNBP}),
\item
NBP-UNIF: a uniform prior on $a$, i.e. $a \sim \mathcal{U}(1/n, 1)$ in (\ref{NBPhierFB}),
\item
NBP-TC: a truncated standard Cauchy prior on $a$, i.e. $\pi(a) = [\textrm{arctan}(1) -\textrm{arctan}(1/n)]^{-1} (1+a)^{-1} \mathbb{I} \{ 1/n < a < 1 \}$ in (\ref{NBPhierFB}). For shorthand notation, we denote this prior as $a \sim \mathcal{TC} (0, 1; 1/n, 1 )$.
\end{enumerate}
 For each of these models, we apply the appropriate thresholding rule: (\ref{thresholdingruleEB}) for NBP-ES, (\ref{thresholdingruleREML}) for NBP-REML, and (\ref{thresholdingruleFB}) for NBP-UNIF and NBP-TC to classify $\theta_i$'s in our model as either signals ($\theta_i \neq 0$) or noise ($\theta_i = 0$). We estimate the average misclassification probability (MP) for these thresholding rules from 100 replicates. 

We compare the performance of our testing procedures to those under the horseshoe (HS), the horseshoe+ (HS+), and the Dirichlet-Laplace (DL) priors.  In the HS and HS+ models, the sparsity parameter $\tau$ is the variance rescaling parameter in (\ref{globallocal}), while in the DL model, the sparsity parameter $\tau$ is the hyperparameter in the Dirichlet prior, $\mathcal{D}(\tau, \ldots, \tau)$. For each of these models, we estimate $\tau$ using either the ES estimator (\ref{EBsparsity}), $\widehat{\tau}^{ES}$, the REML estimator (\ref{remlNBP}), $\widehat{\tau}^{REML}$, or by placing priors on $\tau$, $\tau \sim \mathcal{U}(1/n, 1)$ or $\tau \sim \mathcal{TC} (0, 1; 1/n, 1)$. Implementation for the HS prior is available in the \textsf{R} package \texttt{horseshoe} \footnote{For the method $\tau \sim \mathcal{U}(1/n,1)$, we slightly modify the code in the \texttt{HS.normal.means} function in the \texttt{horseshoe} \textsf{R} package.}, while the methods for the HS+ and DL priors are available in our package \texttt{NormalBetaPrime}. 

 \begin{figure}[t!]
\centering
\includegraphics[width=\textwidth]{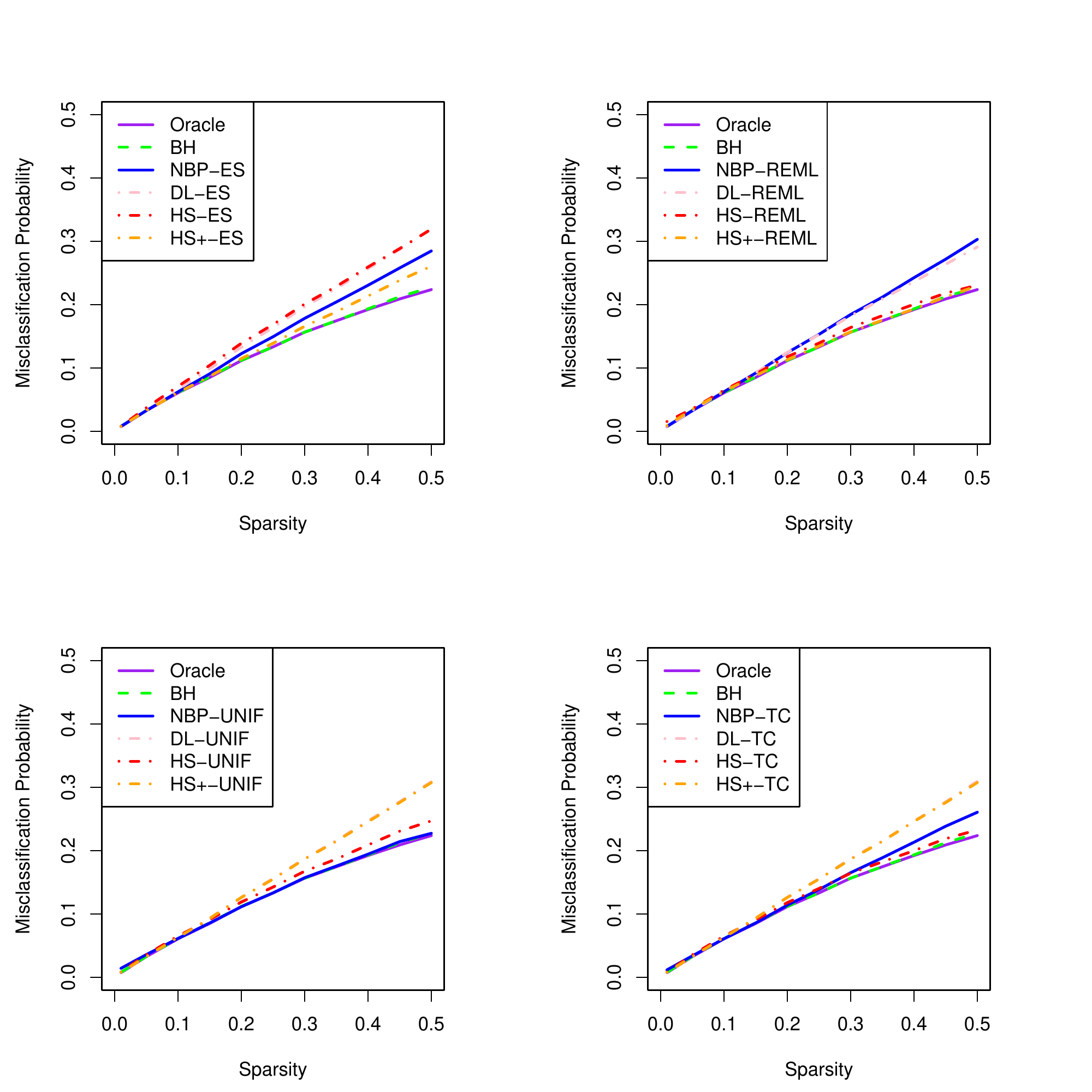} 
\caption{Estimated misclassification probabilities for the NBP, HS, HS+, and DL models when the different estimators for the sparsity parameter are used: estimated sparsity (ES), REML, $\mathcal{U}(1/n, 1)$, and $\mathcal{TC}(0, 1; 1/n, 1)$. The different models are compared to the Bayes Oracle (BO) and Benjamini-Hochberg (BH) procedures.}
\label{fig:2}
\end{figure}

Figure \ref{fig:2} plots the estimated misclassification probabilities (MP) against the true sparsity level $p$ for each of the models, along with the MP's for the Bayes Oracle (BO) and the Benjamini-Hochberg procedure (BH). Recall that the Bayes Oracle rule, defined in (\ref{BayesOracle}), is the decision rule that minimizes the expected number of misclassified signals (\ref{twogroupsrisk}) when $(p, \psi)$ are known. The Bayes Oracle therefore serves as the lower bound to the MP. For the Benjamini-Hochberg rule, we use $\alpha = 1/ \log n = 0.1887$.  \cite{BogdanChakrabartiFrommletGhosh2011} theoretically established for this choice of $\alpha$, the BH procedure is ABOS.

Figure \ref{fig:2} illustrates that all the different models perform very similarly to the Bayes Oracle in sparse situations ($p$ in the range of 0.01 to 0.30), regardless of whether the sparsity parameter is estimated by empirical Bayes or by hierarchical Bayes. Our numerical experiments thus corroborate our theoretical findings that the NBP prior (\ref{NBPhier}) is well-behaved under sparsity. If the ES estimator (\ref{EBsparsity}) is used, then the HS+-ES prior performs the best, with the NBP-ES model following closely behind. If the REML estimator (\ref{remlNBP}) is used, then the NBP-REML model performs well under sparsity but not as well as the other methods in more dense situations. Under the truncated Cauchy prior, the NBP-TC model performs the second best behind HS-TC in dense situations. Finally, under the uniform prior, the NBP-UNIF outperforms all the other models and behaves very similarly to the Bayes Oracle across \textit{all} sparsity levels.  Based on our empirical results, the NBP prior displays the best overall performance when $a$ is endowed with a uniform prior, $a \sim \mathcal{U}(1/n, 1)$.

 \begin{figure}[t!]
\centering
\includegraphics[width=.85\textwidth]{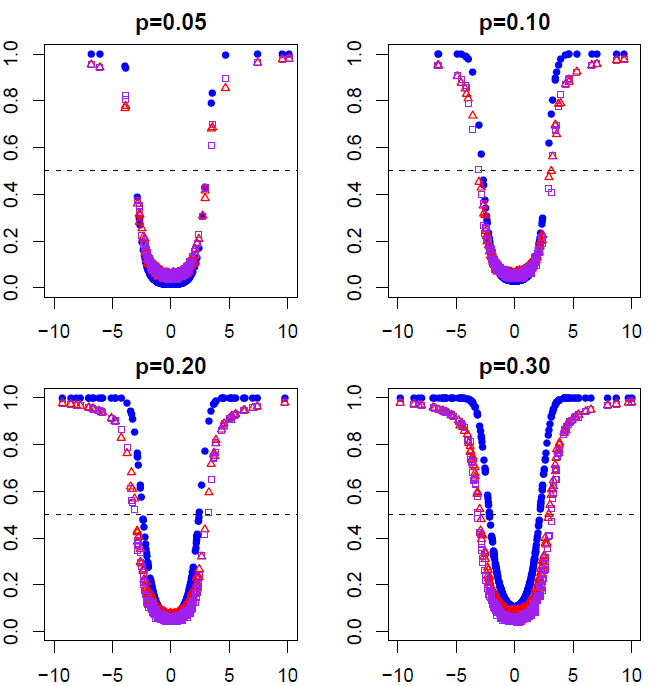} 
\caption{Comparison between the posterior inclusion probabilities $\omega_i(X_i) =  \pi(\nu_i = 1 | X_i)$ and the posterior shrinkage weights $\mathbb{E}(1-\kappa_i | \widehat{a}^{ES}, X_i)$, $\mathbb{E}(1-\kappa_i | \widehat{a}^{REML}, X_i)$. The solid circles are the posterior inclusion probabilities, while the empty triangles correspond to NBP-ES and the empty squares correspond to NBP-REML. }
\label{fig:3}

\end{figure}
 \begin{figure}[h!]
	\centering
	\includegraphics[width=.85\textwidth]{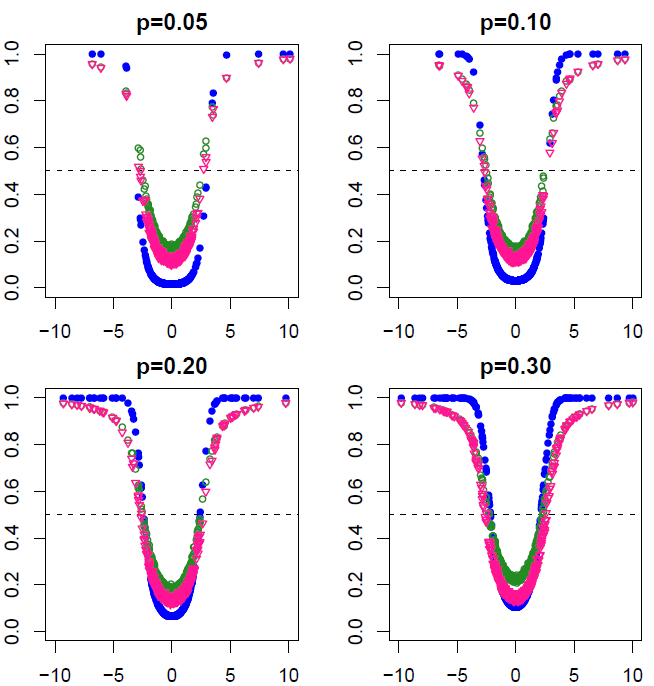} 
	\caption{Comparison between the posterior inclusion probabilities $\omega_i(X_i) =  \pi(\nu_i = 1 | X_i)$ and the posterior shrinkage weights $\mathbb{E}(1-\kappa_i | X_1, \ldots, X_n)$ under the hierarchical Bayes approaches. The solid circles are the posterior inclusion probabilities, while the empty circles correspond to NBP-UNIF and the empty upside-down triangles correspond to NBP-TC. }
	\label{fig:4}
\end{figure}

Taking different choices of $p \in \{0.05, 0.10, 0.20, 0.30 \}$, we plot in Figures \ref{fig:3} and Figure \ref{fig:4} the theoretical posterior inclusion probabilities $\omega_i(X_i) = P(\nu_i = 1 | X_i)$ for the two-groups model (\ref{thetaspikeslab}) given by
\begin{equation*}
\omega_i (X_i) = \pi(\nu_i = 1 | X_i) = \left\{  \left( \frac{1-p}{p} \right) \sqrt{1+ \psi^2} e^{-\frac{X_i^2}{2} \frac{\psi^2}{1+ \psi^2}} + 1 \right\}^{-1},
\end{equation*}
along with the shrinkage weights $\mathbb{E}(1-\kappa_i | \widehat{a}^{ES}, X_i)$, $\mathbb{E}(1-\kappa_i | \widehat{a}^{REML}, X_i)$, and $\mathbb{E}( 1 - \kappa_i | X_1, \ldots, X_n)$ for the NBP-ES, NBP-REML, NBP-UNIF, and NBP-TC models. These plots shows that for small values of the sparsity level $p$, the shrinkage weights are in close proximity to the posterior inclusion probabilities. This offers empirical support for the use of these posterior shrinkage weights  as an approximation to the corresponding posterior inclusion probabilities $\omega_i(X_i)$ in sparse situations.

\begin{figure}[t!]
	\centering
	\includegraphics[width=.64\textwidth]{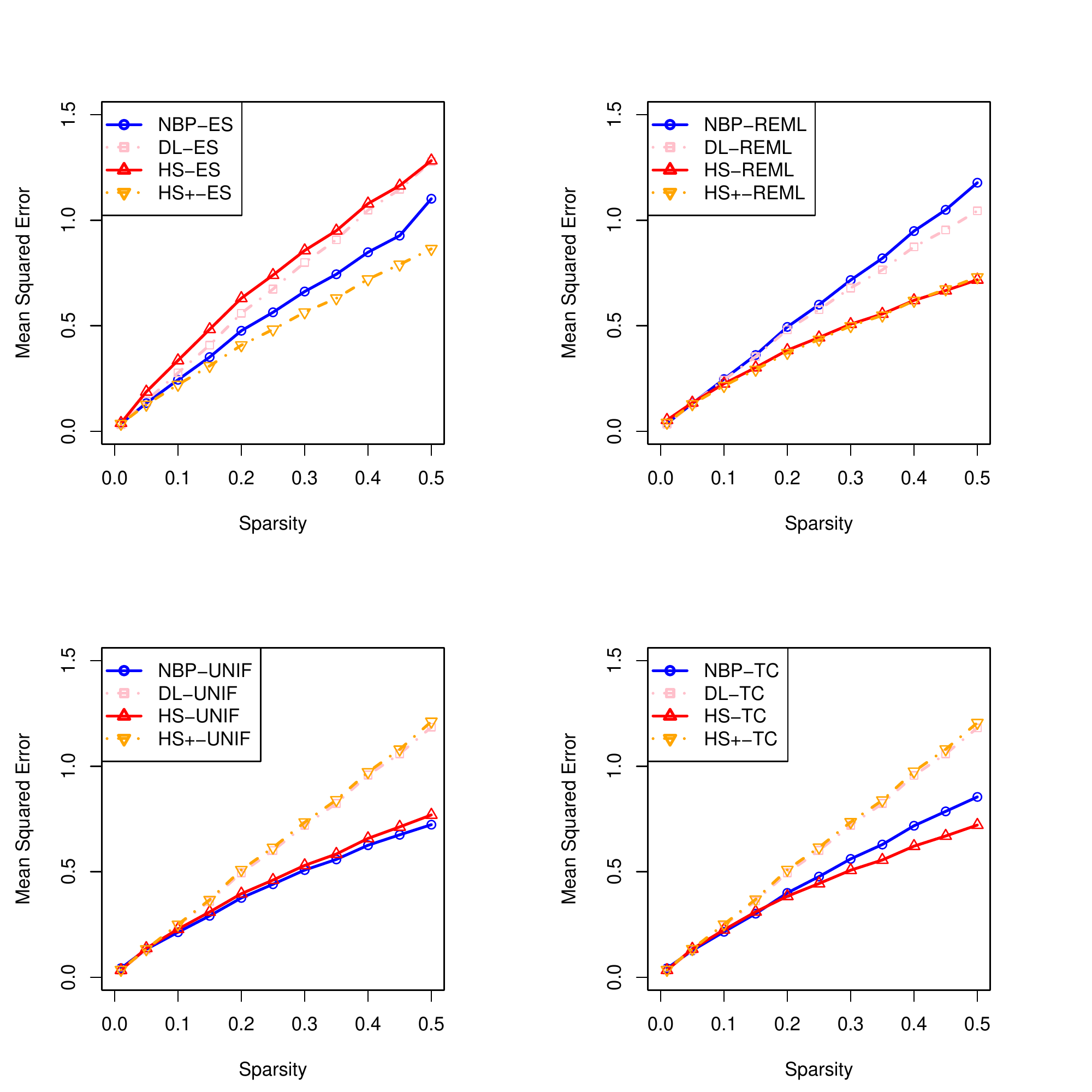} 
	\caption{Mean squared error for the NBP, HS, HS+, and DL models when the different estimators for the sparsity parameter are used: estimated sparsity (ES), REML, $\mathcal{U}(1/n, 1)$, and $\mathcal{TC}(0, 1; 1/n, 1)$. The different models are compared to the Bayes Oracle (BO) and Benjamini-Hochberg (BH) procedures.}
	\label{fig:5}
\end{figure}

\begin{figure}[h!]
	\centering
	\includegraphics[width=.64\textwidth]{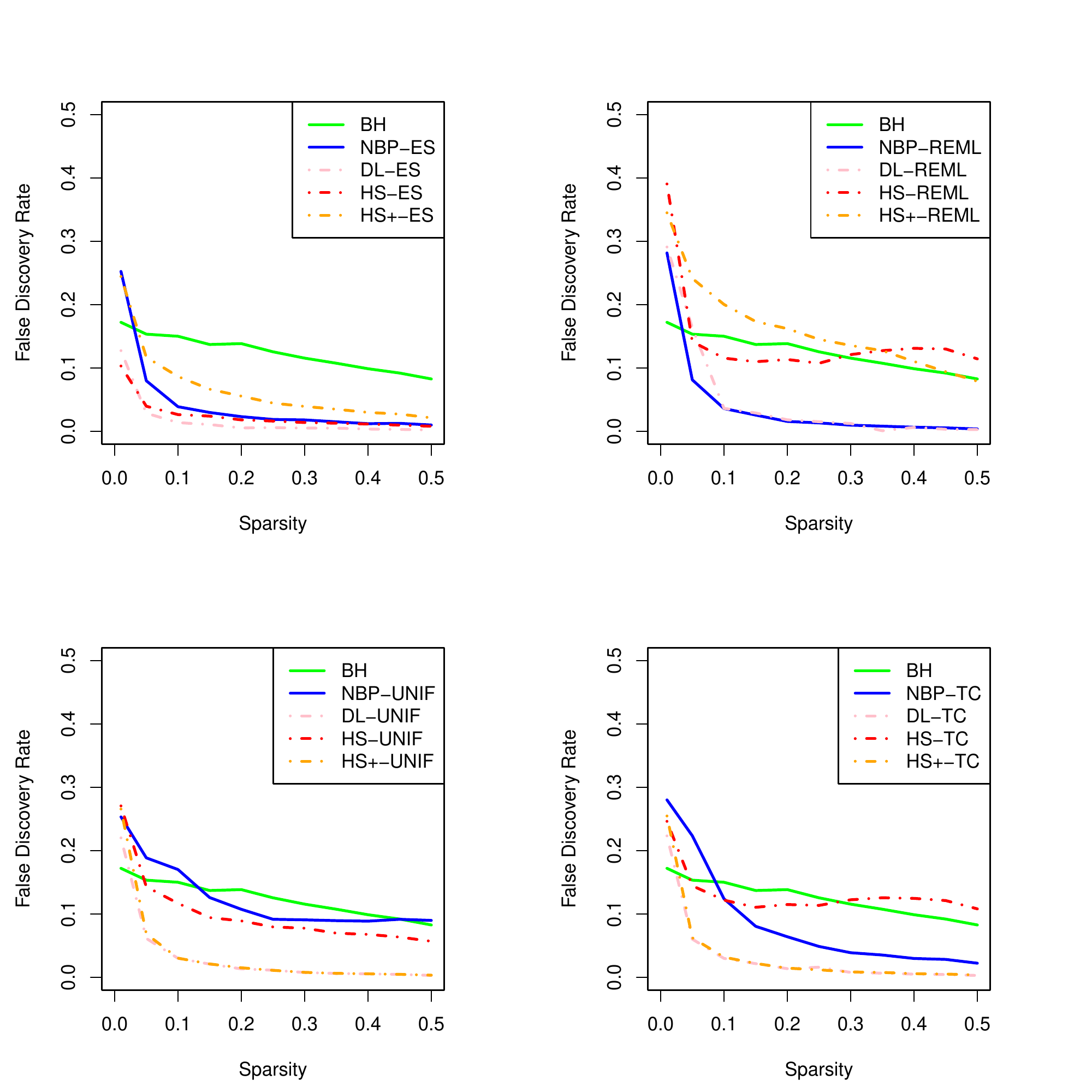} 
	\caption{False discovery rates for the NBP, HS, HS+, and DL models when the different estimators for the sparsity parameter are used: estimated sparsity (ES), REML, $\mathcal{U}(1/n, 1)$, and $\mathcal{TC}(0, 1; 1/n, 1)$. The different models are compared to the Bayes Oracle (BO) and Benjamini-Hochberg (BH) procedures.}
	\label{fig:6}
\end{figure}

\subsection{Estimation and False Discovery Rate (FDR) Control}

While our focus has been on designing a test procedure with the NBP prior which has the Bayes Oracle property, practitioners may also be interested in estimation of the underlying $\bm{\theta}$ or in false discovery rate (FDR) control. Let $\Delta_i$ and $\Omega_i$ be defined as
\begin{equation*}
\begin{array}{c}
\Delta_i \equiv \{ H_{0i} \textrm{ is rejected when } H_{0i} \textrm{ is true} \}, \\
\Omega_i \equiv \{ H_{0i} \textrm{ is rejected when } H_{1i} \textrm{ is true} \}.
\end{array}
\end{equation*}
The (empirical) FDR is defined as
\begin{align} \label{FDRs}
\textrm{FDR} = \frac{ \sum_{i=1}^{n} I (  \Delta_i ) }{ \max \{ 1, \sum_{i=1}^{n} I ( \Delta_i) +  \sum_{i=1}^{n} I( \Omega_i ) \} },
\end{align}
and the goal of frequentist FDR control is to design a test such that $\mathbb{E}(FDR) \leq \alpha$ for a prespecified $\alpha \in (0,1)$. Both estimation and FDR control are separate procedures than the ones considered in this paper and indeed may give conflicting results in terms of `optimality.' For example, \cite{SongCheng2018} proved that any estimator $\bm{\widehat{\theta}}$ which asymptotically has FDR of zero cannot simultaneously obtain the minimax estimation rate. Similarly, a procedure which has the Bayes Oracle property is intended to minimize the \textit{total} expected number of misclassified tests (false positives plus false negatives). It is thus conceivable that a test which has high FDR could still be ABOS, provided that the number of false negatives is very low. Conversely, a test which has very low FDR may still have a very high misclassification probability (MP) if the test results in a high number of false negatives.

Thresholding rules (\ref{thresholdingrule}), (\ref{thresholdingruleEB}), (\ref{thresholdingruleREML}), and (\ref{thresholdingruleFB}) are explicitly designed to minimize the expected \textit{total} number of misclassified tests. Nevertheless, it is worth investigating the estimation quality under the NBP prior and the extent to which these tests control the FDR in our simulation study. To assess the estimation of $\bm{\theta}$ under the NBP prior, we compute the mean squared error (MSE) about the posterior median, i.e. $\textrm{MSE} = (1/n) \sum_{i=1}^{n} ( \widehat{\theta_i}^{\textrm{med}} - \theta_{0i})^2$, for the NBP-ES, NBP-REML, NBP-UNIF, and NBP-TC models, averaged across 100 replications. We compare the performance to the respective DL, HS, and HS+ models. Our results are plotted in Figure \ref{fig:5}. Figure \ref{fig:5} shows that the hierarchical Bayes approaches give the best estimation quality for the NBP prior, with the NBP-UNIF prior outperforming all other methods.

We also plot the FDR (\ref{FDRs}) for all our models in Figure \ref{fig:6}. Figure \ref{fig:6} shows that testing rules (\ref{thresholdingruleEB}), (\ref{thresholdingruleREML}), and (\ref{thresholdingruleFB}) under the NBP, DL, HS, and HS+ priors all control FDR well. For most of the sparsity levels, the FDRs under these shrinkage priors are lower than the FDR under BH. In particular, Figure \ref{fig:6} shows that tests under the plug-in ES and REML estimators give FDR close to zero in dense settings. However, as illustrated in Figures \ref{fig:2} and \ref{fig:5}, NBP-REML has the highest \textit{total} misclassification rate and estimation error, indicating that NBP-REML misses a large proportion of actual signals in dense settings. Based on our numerical studies, we recommend the hierarchical Bayes NBP prior with $a \sim \mathcal{U}(1/n, 1)$ as the `default' implementation for the NBP model. Figure \ref{fig:6} shows that the FDR under the NBP-UNIF model compares favorably to that of the BH procedure. In addition, NBP-UNIF mimics the Bayes Oracle performance the closest and has the lowest estimation error. 

If a more conservative test is desired, then we recommend using the NBP-ES model. The NBP-ES model performs slightly worse than NBP-UNIF in terms of MP and MSE, but it has lower FDR. At present, designing theoretically rigorous tests with frequentist FDR control using scale-mixture shrinkage priors (\ref{scalemixture}) is still an open problem.

\section{Analysis of a Prostate Cancer Data Set} \label{DataAnalysis}
We demonstrate practical application of the NBP prior using a popular prostate cancer data set introduced by \cite{Singhetal2002}. In this data set, there are gene expression values for $n = 6033$ genes for $m = 102$ subjects, with $m_1 = 50$ normal control subjects and $m_2 = 52$ prostate cancer patients. We aim to identify genes that are significantly different between control and cancer patients. We first conduct a two-sample t-test for each gene and then transform the test statistics $(t_1, ..., t_n)$ to z-scores using the inverse normal cumulative distribution function (CDF) transform $\Phi^{-1}(F_{t_{100}}(t_i))$, where $F_{t_{100}}$ denotes the CDF for the Student's t distribution with 100 degrees of freedom. 

With z-scores $(z_1, ..., z_n)$, it is clear that $z_i$ follows a standard normal distribution under the null hypothesis, i.e. $H_{0i}: z_i \sim \mathcal{N}(0,1), i = 1, \ldots, n$. This allows us to implement the NBP  prior on the z-scores to conduct simultaneous testing of $H_{0i}: \theta_i = 0$ vs. $H_{1i}: \theta_i \neq 0$, $i = 1, ..., n,$ to identify genes that are significantly associated with prostate cancer. Additionally, we can also estimate $\bm{\theta} = (\theta_1, ..., \theta_n)$ under model (\ref{X=theta+eps}) using the posterior mean. As argued by \cite{Efron2010}, $| \theta_i |$ can be interpreted as the effect size of the $i$th gene for prostate cancer. \cite{Efron2010} first analysed this model for this particular data set by obtaining empirical Bayes estimates $\widehat{\theta}_i^{Efron}$, $i = 1, ...,n$,  based on the two-groups model (\ref{thetaspikeslab}). In our analysis, we use the posterior means $\widehat{\theta}_i, i = 1, ..., n,$ to estimate the strength of association.

We implement the NBP-UNIF model and use classification rule (\ref{thresholdingruleFB}) to identify significant genes. For comparison, we also fit this model for the DL-UNIF, HS-UNIF, and HS+-UNIF priors, and benchmark these models to the Benjamini-Hochberg (BH) procedure with FDR $\alpha= 0.10$. The NBP-UNIF model selects 165 out of the 6033 genes as significant, in comparison to 60 genes under the BH procedure. All 60 genes selected by the BH procedure are included in the 166 genes determined to be significant by the NBP prior. The HS-UNIF and HS+-UNIF priors select 55 and 38 genes respectively as significant, while the DL prior selects 102 genes as significant.

Table \ref{Table:3} shows the top 10 genes selected by \cite{Efron2010} and their estimated effect size on prostate cancer. We compare \cite{Efron2010}'s empirical Bayes posterior mean estimates with the posterior mean estimates under the NBP-UNIF, DL-UNIF, HS-UNIF, and HS+-UNIF priors. Our results confirm the tail robustness of the NBP prior. All of the scale-mixture shrinkage priors shrink the estimated effect size for significant genes less aggressively than Efron's procedure. On this particular dataset, the NBP model shrinks large signals the least of all the methods considered when the sparsity parameter $a$ is endowed with a prior, $a \sim \mathcal{U}(1/n, 1)$.

Figure \ref{fig:6} illustrated that the hierarchical Bayes model with a uniform prior tends to give higher FDR than the models where the sparsity parameter is estimated with the estimated sparsity (ES) plug-in estimator $\widehat{a}^{ES}$ (\ref{EBsparsity}). In some applications, it may be better to have tests which are more conservative. With this in mind, we repeat our analysis using (\ref{EBsparsity}) as the sparsity parameter and classification rule (\ref{thresholdingruleEB}). In this case, the NBP-ES model selected just 72 genes, including all 60 genes selected by the BH procedure. The DL-ES and HS-ES models were also more conservative, selecting 39 and 4 genes respectively. The HS+-ES model selected 50 genes as significant.

\begin{table}[t!]  
	\centering
	\caption{The z-scores and the effect size estimates for the top 10 genes selected by \cite{Efron2010} by the NBP-UNIF, DL-UNIF, HS-UNIF, and HS+-UNIF models and the two-groups empirical Bayes model by \cite{Efron2010}.}  
	\begin{tabular}{c c c c c c c}
		Gene & z-score & $\widehat{\theta}_i^{NBP}$ & $\widehat{\theta}_i^{DL}$ & $\widehat{\theta}_i^{HS}$ &  $\widehat{\theta}_i^{HS+}$ & $\widehat{\theta}_i^{Efron}$ \\ 
		\hline
		610 & 5.29 & 4.87 & 4.61 & 4.87 & 4.87 & 4.11 \\
		1720 & 4.83 & 4.39 & 4.09 & 4.30 & 4.37 & 3.65 \\ 
		332 & 4.47 & 3.97 & 3.62 & 3.85 & 3.73 & 3.24 \\
		364 & -4.42 & -3.94 & -3.56 & -3.81 & -3.85 & -3.57 \\
		914 & 4.40 & 3.85 & 3.54 & 3.74 & 3.71 & 3.16 \\
		3940 & -4.33 & -3.80 & -3.49 & -3.53 & -3.68 & -3.52 \\  
		4546 & -4.29 & -3.74 & -3.39 & -3.58 & -3.70 & -3.47 \\
		1068 & 4.25 & 3.69 & 3.31 & 3.41 & 3.35 & 2.99 \\
		579 & 4.19 & 3.60 & 3.32 & 3.38 & 3.43 & 2.92 \\
		4331 & -4.14 & -3.54 & -3.14 & -3.23 & -3.19 & -3.30 \\
		\hline
	\end{tabular}
	\label{Table:3}
\end{table}

\section{Concluding Remarks and Future Work}
In this paper, we have studied a scale-mixture shrinkage prior with the beta prime prior (\ref{TPBN}) as the scale parameters for multiple testing under sparsity. By appropriately estimating the sparsity parameter in the normal-beta prime prior and thresholding the posterior shrinkage weight, the NBP can be used to identify signals in sparse normal mean vectors. We have investigated these testing rules within the decision theoretic framework of \cite{BogdanChakrabartiFrommletGhosh2011} and established that the NBP prior has the Bayes Oracle property.

Our results also suggest that scale-mixture shrinkage priors of the most general form (\ref{scalemixture}) can asymptotically attain the exact optimal Bayes risk for multiple testing. In the future, we hope to derive general sufficient conditions under which shrinkage priors (\ref{scalemixture}) are asymptotically Bayes optimal under sparsity. We would also like to provide theoretical justification for the use of the restricted marginal maximum likelihood (REML) and hierarchical Bayes methods presented in Section \ref{DataAdaptiveMethods}. Previously, \cite{VanDerPasSzaboVanDerVaart2017} showed that these adaptive methods lead to near-minimax \textit{estimation} under the horseshoe prior. Our results suggest that these methods are also optimal for multiple testing and that they are appropriate to use for general shrinkage priors besides the horseshoe.

Finally, there has been a rapid growth in the `frequentist Bayes' theory field in recent years, but the literature on frequentist assessment of Bayesian multiple testing procedures is only now emerging. In a recent preprint, \cite{CastilloRoquain2018} show that thresholding the posterior under a point-mass spike-and-slab prior at level $\alpha \in (0, 1)$ asymptotically gives frequentist false discovery rate (FDR) control of level $\alpha$ (up to a multiplicative constant) for sparse normal means. We conjecture that thresholding rules based on the posterior shrinkage weight under the NBP prior (\ref{NBPhier}) -- and under general scale-mixture shrinkage priors (\ref{scalemixture}) -- can also be constructed for frequentist FDR control.

\section*{Acknowledgments}

The authors would like to thank Dr. Anirban Bhattacharya and Dr. Xueying Tang for sharing their codes, which were modified to generate Figures \ref{fig:1}-\ref{fig:6}. We are grateful to two anonymous reviewers and the Associate Editors whose thoughtful comments and suggestions helped to greatly improve this paper.

\section*{Disclosure statement}

The authors declare that we have no conflicts of interest in the authorship or publication of this contribution.

\section*{Supplementary Data}

The Supplementary Materials document contains the proofs for the propositions and theorems in Sections \ref{ConcentrationProperties}, \ref{NBPTestingNonAdaptive}, and \ref{NBPTestingAdaptive}, as well as the technical details for implementing our model. Code to implement our model is available in the \textsf{R} package \texttt{NormalBetaPrime}, which also contains the prostate cancer data set analysed in Section \ref{DataAnalysis}.

\appendix

\section{Proofs for Section 2.1} \label{App:A}

\begin{proof}[Proof of Proposition \ref{Prop:1}]
	As noted by Proposition 1 in \cite{ArmaganClydeDunson2011}, the beta prime density (\ref{TPBN}) can be rewritten as a product of independent gamma and inverse gamma densities. We thus reparametrize model (\ref{NBPhier}) for a single observation $\theta$ as follows:
	\begin{equation} \label{NBPhiersingle}
	\begin{array}{c}
	\theta | \lambda_i \xi \sim \mathcal{N}(0, \lambda \xi),  \\
	\lambda \sim \mathcal{G} (a, 1),  \\
	\xi \sim \mathcal{IG} (b, 1).
	\end{array}
	\end{equation} 
	From (\ref{NBPhiersingle}), we see that the joint distribution of the prior is proportional to
	\begin{align*}
	\pi( \theta, \lambda, \xi ) & \propto  \left( \lambda \xi \right)^{-1/2} \exp \left( - \frac{\theta^2}{2 \lambda \xi} \right)  \lambda^{a - 1} \exp \left( - \lambda \right) \exp \left( - \frac{1}{\xi} \right) \xi^{-b - 1}  \\
	&  = \lambda^{a - 3/2} \exp(- \lambda) \xi^{-b - 3/2}  \exp \left( - \left( \frac{\theta^2}{2 \lambda} + 1 \right) \frac{1}{\xi} \right).
	\end{align*}
	Thus, 
	\begin{align*}
	\pi (\theta, \lambda) & \propto \lambda^{a - 3/2} \exp(- \lambda) \displaystyle \int_{\xi = 0}^{\infty} \xi^{-b - 3/2} \exp \left( - \left( \frac{\theta^2}{2 \lambda} + 1 \right) \frac{1}{\xi} \right) d \xi \\
	& \propto \left( \frac{\theta^2}{2 \lambda} + 1 \right)^{- (b + 1/2)} \lambda^{a - 3/2} e^{- \lambda},
	\end{align*}
	and thus, the marginal density of $\theta$ is proportional to
	\begin{equation}  \label{thetamarginal}
	\pi (\theta) \propto \displaystyle \int_{0}^{\infty} \left( \frac{\theta^2}{2 \lambda} + 1 \right)^{- (b + 1/2)} \lambda^{a - 3/2} e^{- \lambda} d \lambda.
	\end{equation} 
	As $|\theta| \rightarrow 0$, the expression in (\ref{thetamarginal}) is bounded below by
	\begin{equation} \label{thetalowerbound}
	C \displaystyle \int_{0}^{\infty} \lambda^{a - 3/2} e^{- \lambda}  d \lambda,
	\end{equation} 
	where $C$ is a constant that depends on $a$ and $b$. The integral expression in (\ref{thetalowerbound}) clearly diverges to $\infty$ for any $0 < a \leq 1/2$. Therefore, (\ref{thetamarginal}) diverges to infinity as $| \theta | \rightarrow 0$, by the monotone convergence theorem.
\end{proof}

\begin{proof}[Proof of Theorem \ref{Th:1}]
	From (\ref{kappadensity}), the posterior distribution of $\kappa_i$ under $\textrm{NBP}_n$ is proportional to
	\begin{equation} \label{NBPnkappa}
	\pi(\kappa_i | X_i) \propto \exp \left( - \frac{\kappa_i X_i^2}{2} \right) \kappa_i^{b - 1/2} (1 - \kappa_i)^{a_n - 1}, \hspace{.3cm} \kappa_i \in (0, 1).
	\end{equation} 
	Hence,
	\begin{align*}
	\mathbb{E} ( 1 - \kappa_i | X_i ) & = \frac{ \displaystyle \int_{0}^{1} \kappa_i^{b - 1/2} (1 - \kappa_i)^{a_n} \exp \left( - \frac{\kappa_i X_i^2}{2} \right) d \kappa_i }{ \displaystyle \int_{0}^{1} \kappa_i^{b - 1/2} (1 - \kappa_i)^{a_n - 1} \exp \left( - \frac{\kappa_i X_i^2}{2} \right) d \kappa_i }   \\
	& \leq \frac{ e^{X_i^2/2} \displaystyle \int_{0}^{1} \kappa_i^{b - 1/2} (1 - \kappa_i)^{a_n} d \kappa_i }{ \displaystyle \int_{0}^{1} \kappa_i^{b - 1/2} ( 1 - \kappa_i)^{a_n - 1} d \kappa_i } \\
	& = e^{X_i^2/2} \frac{ \Gamma (a_n + 1 ) \Gamma (b + 1/2) }{\Gamma (a_n + b + 3/2) } \times \frac{ \Gamma (a_n + b + 1/2)}{ \Gamma (a_n) \Gamma (b + 1/2)  } \\
	& = e^{X_i^2/2} \left( \frac{a_n }{a_n + b + 1/2} \right).
	\end{align*}
\end{proof}

\begin{proof}[Proof of Theorem \ref{Th:2}]
	Note that since $b \in ( \frac{1}{2}, \infty)$, $\kappa_i^{b - 1/2}$ is increasing in $\kappa_i$ on $(0, 1)$. Additionally, since $a_n \in (0, 1 )$, $(1 - \kappa_i)^{a_n - 1}$ is increasing in $\kappa_i$ on $(0, 1)$. Using these facts, we have
	\begin{align*}
	\Pr ( \kappa_i < \epsilon | X_i ) & \leq \frac{ \displaystyle \int_{0}^{\epsilon} \exp \left( - \frac{\kappa_i X_i^2}{2} \right) \kappa_i^{b - 1/2} (1 - \kappa_i)^{a_n - 1} d \kappa_i}{ \displaystyle \int_{\epsilon}^{1} \exp \left( - \frac{\kappa_i X_i^2}{2} \right) \kappa_i^{b - 1/2} (1 - \kappa_i)^{a_n - 1} d \kappa_i }  
	\end{align*}
	\begin{align*}
	& \leq \frac{ e^{X_i^2/2} \displaystyle \int_{0}^{\epsilon} \kappa_i^{b - 1/2} (1 - \kappa_i)^{a_n - 1} d \kappa_i}{ \displaystyle \int_{\epsilon}^{1} \kappa_i^{b - 1/2} (1 - \kappa_i)^{a_n - 1} d \kappa_i }  \\
	& \leq \frac{ e^{X_i^2/2} (1 - \epsilon)^{a_n - 1} \displaystyle \int_{0}^{\epsilon} \kappa_i^{b - 1/2} d \kappa_i }{ \epsilon^{b - 1/2} \displaystyle \int_{\epsilon}^{1} (1 - \kappa_i)^{a_n - 1} d \kappa_i }  \\
	& = \frac{ e^{X_i^2/2}  (1 - \epsilon)^{a_n - 1} \left( b + \frac{1}{2} \right)^{-1} \epsilon^{b+1/2} }{ a_n^{-1} \epsilon^{b - 1/2}  (1 - \epsilon)^{a} }    \\
	& =  e^{X_i^2/2} \frac{ a_n \epsilon }{ \left( b + 1/2 \right) (1 - \epsilon)}.
	\end{align*}
\end{proof}

\begin{proof}[Proof of Theorem \ref{Th:3}]
	Letting $C$ denote the normalizing constant, we have
	\begin{align*}
	\displaystyle\int_{0}^{\eta} \pi (\kappa_i | X_i) d \kappa_i & = C \displaystyle \int_{0}^{\eta} \exp \left( - \frac{\kappa_i X_i^2}{2} \right) \kappa_i^{b - 1/2 } (1 - \kappa_i)^{a_n - 1} d \kappa_i \\
	& \geq C \displaystyle \int_{0}^{\eta \delta} \exp \left( -  \frac{\kappa_i X_i^2}{2} \right) \kappa_i^{b - 1/2} (1 - \kappa_i)^{a_n - 1} d \kappa_i \\
	& \geq C \exp \left( - \frac{\eta \delta}{2} X_i^2 \right) \displaystyle \int_{0}^{\eta \delta} \kappa_i^{b - 1/2} d \kappa_i \\
	& = C \exp \left( -\frac{\eta \delta}{2} X_i^2 \right) \left( b + \frac{1}{2} \right)^{-1} ( \eta \delta )^{b+\frac{1}{2}}. \numbereqn \label{pikappa0n}
	\end{align*}
	Also, since $b \in (\frac{1}{2}, \infty)$, $\kappa_i^{b - 1/2}$ is increasing in $\kappa_i$ on $(0, 1)$.
	\begin{align*}
	\displaystyle \int_{\eta}^{1} \pi (\kappa_i | X_i ) d \kappa_i & = C \displaystyle \int_{\eta}^{1} \exp \left( - \frac{\kappa_i X_i^2}{2} \right) \kappa_i^{b - 1/2} (1 - \kappa_i)^{a_n - 1} d \kappa_i \\
	& \leq C \exp \left( - \frac{\eta X_i^2}{2} \right) \displaystyle \int_{\eta}^{1} \kappa_i^{b - 1/2} (1 - \kappa_i)^{a_n - 1} d \kappa_i \\
	& \leq C \exp \left( - \frac{\eta X_i^2}{2} \right) \displaystyle \int_{\eta}^{1} (1 - \kappa_i)^{a_n - 1} d \kappa_i \\
	& = C \exp \left( - \frac{ \eta X_i^2}{2} \right) a_n^{-1} (1 - \eta)^{a_n}. \numbereqn \label{pikappan1}
	\end{align*}
	Combining (\ref{pikappa0n}) and (\ref{pikappan1}), we have
	\begin{align*}
	\Pr ( \kappa_i > \eta | X_i) & \leq \frac{ \displaystyle \int_{\eta}^{1} \pi (\kappa_i | X_i) d \kappa_i }{ \displaystyle \int_{0}^{\eta} \pi(\kappa_i | X_i ) d \kappa_i } \leq  \frac{ \left( b + \frac{1}{2} \right) (1 - \eta)^{a_n} }{ a_n ( \eta \delta)^{b + \frac{1}{2}} } \exp \left( - \frac{\eta (1 - \delta)}{2} X_i^2 \right).
	\end{align*}
\end{proof}

\section{Proofs for Section 3.3} \label{App:B}
Our proof methods follow those of \cite{DattaGhosh2013, GhoshTangGhoshChakrabarti2016, GhoshChakrabarti2017}, except our arguments rely on control of the sequence of hyperparameters $a_n$, rather than on specifying a rate or an estimate for a rescaling parameter $\tau$, as in the class of priors (\ref{globallocal}). Moreover, we make explicit use of Theorems 2.1-2.3 in the present manuscript in our proofs.

\begin{proof}[Proof of Theorem \ref{Theorem:4}]
	By Theorem \ref{Th:1}, the event $\left\{ \mathbb{E} ( 1 - \kappa_i | X_i) > \frac{1}{2} \right\}$ implies the event
	\begin{align*}
	& \left\{ e^{X_i^2/2} \left( \frac{a_n }{a_n + b + 1/2} \right) > \frac{1}{2} \right\} \hspace{.2cm} \Leftrightarrow \hspace{.2cm} \left\{ X_i^2 > 2 \log \left( \frac{a_n + b + 1/2}{2 a_n} \right) \right\}.
	\end{align*}
	Therefore, noting that under $H_{0i}$, $X_i \sim \mathcal{N}(0, 1)$ and using Mill's ratio, i.e. $P( |Z| > x) \leq \frac{2 \phi (x)}{x}$, we have 
	\begin{align*} \label{typeIupperbound}
	t_{1i} & \leq \Pr \left( X_i^2 > 2 \log \left( \frac{a_n + b + 1/2}{2 a_n} \right) \bigg| H_{0i} \textrm{ is true } \right) \\
	& = \Pr \left( |Z| > \sqrt{2 \log \left( \frac{a_n + b + 1/2}{2 a_n } \right)} \right) \\
	& \leq \frac{ 2 \phi \left(  \sqrt{2 \log \left( \frac{a_n + b + 1/2}{2 a_n } \right)} \right)}{  \sqrt{2 \log \left( \frac{a_n + b + 1/2}{2 a_n } \right)}} \\
	& = \frac{2 \sqrt{2} a_n }{\sqrt{\pi} (a_n + b + 1/2)} \left[ \log \left( \frac{a_n + b + 1/2}{2 a_n } \right) \right]^{-1/2}. \numbereqn
	\end{align*}
\end{proof}
\begin{proof}[Proof of Theorem \ref{Theorem:5}]
	By definition, the probability of a Type I error for the $i$th decision is given by
	\begin{equation*}
	t_{1i} = \Pr \left[ \mathbb{E} (1 - \kappa_i | X_i) > \frac{1}{2} \bigg| H_{0i} \textrm{ is true} \right].
	\end{equation*}
	Fix $\xi \in (0, 1/2)$. By Theorem \ref{Th:3},
	\begin{equation*}
	\mathbb{E} (\kappa_i | X_i) \leq \xi + \frac{ \left( b + \frac{1}{2} \right) (1 - \xi)^{a_n} }{ a_n ( \xi \delta)^{b + \frac{1}{2}} } \exp \left( - \frac{\xi (1 - \delta)}{2} X_i^2 \right).
	\end{equation*}
	Hence,
	\begin{equation*}
	\left\{ \mathbb{E}(1 - \kappa_i | X_i) > \frac{1}{2} \right\} \supseteq \left\{ \frac{ \left( b + \frac{1}{2} \right) (1 - \xi)^{a_n} }{ a_n ( \xi \delta)^{b + \frac{1}{2}} } \exp \left( - \frac{\xi (1 - \delta)}{2} X_i^2 \right) < \frac{1}{2} - \xi \right\}. 
	\end{equation*}
	Thus, using the definition of $t_{1i}$ and noting that under $H_{0i}$, $X_i \sim \mathcal{N}(0,1)$, as $n \rightarrow \infty$,
	
	\begin{align*}
	t_{1i} & \geq \Pr \left(  \frac{ \left( b + \frac{1}{2} \right) (1 - \xi)^{a_n} }{ a_n ( \xi \delta)^{b + \frac{1}{2}} }  \exp \left( - \frac{\xi (1 - \delta)}{2} X_i^2 \right) < \frac{1}{2} - \xi \hspace{.1cm} \bigg| \hspace{.1cm} H_{0i} \textrm{ is true} \right) \\
	& = \Pr \left( X_i^2 > \frac{2}{\xi (1 - \delta)} \left[ \log \left( \frac{\left(b + \frac{1}{2} \right) \left( 1 - \xi \right)^{a_n} }{a_n (\xi \delta)^{b + \frac{1}{2}} \left( \frac{1}{2} - \xi \right) } \right)  \right] \right) \\
	&  = 2 \Pr \left( Z > \sqrt{ \frac{2}{\xi (1 - \delta)} \left[ \log \left( \frac{\left(b + \frac{1}{2} \right) \left( 1 - \xi \right)^{a_n} }{a_n (\xi \delta)^{b + \frac{1}{2}} \left( \frac{1}{2} - \xi \right) } \right)  \right] } \right) \\
	& = 2 \left(1 - \Phi \left( \sqrt{ \frac{2}{\xi (1 - \delta)} \left[ \log \left( \frac{\left(b + \frac{1}{2} \right) \left( 1 - \xi \right)^{a_n} }{a_n (\xi \delta)^{b + \frac{1}{2}} \left( \frac{1}{2} - \xi \right) } \right)  \right]  } \right) \right),
	\end{align*}
	where for the second to last inequality, we used the fact that $a_n \rightarrow 0$ as $n \rightarrow \infty$, and the fact that both $\xi$ and $\xi \delta \in (0, \frac{1}{2})$, so that the $\log ( \cdot )$ term in final equality is greater than zero for sufficiently large $n$.
\end{proof}

Before proving the asymptotic upper bound on Type II error in Theorem \ref{Theorem:6}, we first prove a lemma that bounds the quantity $\mathbb{E}(\kappa_i | X_i)$ from above for a single $X_i$.    
\begin{lemma}
	\label{Lemma:1}
	Suppose we observe $\bm{X} \sim \mathcal{N} ( \bm{\theta}, \bm{I}_n)$ and we place an $\textrm{NBP}_n$ prior (\ref{NBPn}) on $\bm{\theta}$, with and $a_n \in (0, 1)$ where $a_n \rightarrow 0$ as $n \rightarrow \infty$, and fixed $b \in (1/2, \infty)$. Fix constants $\eta \in (0, 1)$, $\delta \in (0, 1)$, and $d > 2$. Then for a single observation $x$ and any $n$, the posterior shrinkage coefficient $\mathbb{E}(\kappa | x)$ can be bounded above by a measurable, non-negative real-valued function $h_n(x)$, given by
	\begin{equation} \label{hnx}
	h_n(x) =  \left\{ 
	\begin{array}{ll}
	C_{n, \eta} \left[ x^2 \int_{0}^{x^2} t^{b-1/2} e^{-t/2} dt \right]^{-1} + \frac{ \left( b + \frac{3}{2} \right)^{-1} (1 - \eta)^{a_n}}{a_n (\eta \delta)^{b+3/2}} \exp \left( - \frac{\eta(1 - \delta)}{2} x^2 \right), & \textrm{if } |x| > 0, \\
	1, & \textrm{if } x = 0,\\ 
	\end{array}
	\right. 
	\end{equation}
	where $C_{n, \eta} = (1 - \eta)^{a_n - 1} \Gamma \left( b+\frac{3}{2} \right) 2^{b+3/2}$. For any $\rho > \frac{2}{\eta (1-\delta)}$, $h_n(x)$ also satisfies
	\begin{equation} \label{hnxlimsup}
	\displaystyle \lim_{n \rightarrow \infty} \displaystyle \sup_{|x| > \sqrt{ \rho \log \left( \frac{1}{a_n} \right)}} h_n(x) = 0.
	\end{equation}
\end{lemma}
\begin{proof}[Proof of Lemma \ref{Lemma:1}]
	We first focus on the case where $|x| > 0$. Fix $\eta \in (0, 1), \delta \in (0, 1)$, and observe that
	\begin{equation} \label{EVkappa}
	\mathbb{E}(\kappa | x) = \mathbb{E} (\kappa 1 \{ \kappa < \eta \} | x) + \mathbb{E} ( \kappa 1 \{ \kappa \geq \eta \} | x).
	\end{equation}
	We consider the two terms in (\ref{EVkappa}) separately. To bound the first term, we have from (9) and the fact that $(1-\kappa)^{a_n-1}$ is increasing in $\kappa \in (0,1)$ for $a_n \in (0,1)$ that
	\begin{align*}
	\mathbb{E}(\kappa 1\{ \kappa < \eta \})  & = \frac{\int_{0}^{\eta} \kappa \cdot \kappa^{b-1/2} (1-\kappa)^{a_n-1} e^{-\kappa x^2/2} d \kappa}{\int_{0}^{1} \kappa^{b-1/2} (1-\kappa)^{a_n-1} e^{-\kappa x^2/2} d \kappa }  \\
	& \leq  (1-\eta)^{a_n-1} \frac{\int_{0}^{\eta} \kappa^{b+1/2} e^{-\kappa x^2/2} d \kappa}{\int_{0}^{1} \kappa^{b-1/2} e^{-\kappa x^2/2} d \kappa}  
	\end{align*}
	\begin{align*}
	& = (1-\eta)^{a_n -1} \frac{1}{x^2} \frac{\int_{0}^{\eta x^2} t^{b+1/2} e^{-t/2} dt}{\int_{0}^{x^2} t^{b-1/2} e^{-t/2} dt} \\
	& \leq (1 - \eta)^{a_n-1} \frac{1}{x^2} \frac{\int_{0}^{\infty} t^{b+1/2} e^{-t/2} dt}{\int_{0}^{x^2} t^{b-1/2} e^{-t/2} dt} \\
	& =  C(n) \left[ x^2  \int_{0}^{x^2} t^{b-1/2} e^{-t/2} dt \right]^{-1}  \\
	& := h_1 (x) \hspace{.3cm} (\textrm{say}), \numbereqn \label{h1x}
	\end{align*}
	where we use a change of variables $t = \kappa x^2$ in the second equality, and $C(n) = (1-\eta)^{a_n-1} \Gamma \left( b+\frac{3}{2} \right) 2^{b+3/2} $.
	
	To bound the second term in (\ref{EVkappa}) from above, we follow the same steps as the proof of Theorem \ref{Th:3}, except we replace $\kappa_i^{b-1/2}$ in the numerators of the integrands with $\kappa_i^{b+1/2}$ to obtain an upper bound,
	\begin{equation} \label{h2x}
	\frac{ \left( b+ \frac{3}{2} \right) (1 - \eta)^{a_n}}{a_n (\eta \delta)^{b+3/2}} \exp \left( - \frac{\eta(1 - \delta)}{2} x^2 \right) := h_2(x) \hspace{.3cm} (\textrm{say}).
	\end{equation}
	
	Combining (\ref{EVkappa})-(\ref{h2x}), we set $h_n(x) = h_1(x) + h_2(x)$ for any $|x| > 0$, and we easily see that for any $x \neq 0$ and fixed $n$, $\mathbb{E}(\kappa | x) \leq h_n (x)$. On the other hand, if $x = 0$, then
	\begin{align*} 
	\mathbb{E} (\kappa | x) & = \frac{\int_{0}^{1} \kappa^{b+1/2} (1-\kappa)^{a_n-1} d \kappa}{\int_{0}^{1} \kappa^{b-1/2} (1-\kappa)^{a_n-1} d \kappa } = \frac{b+1/2}{a_n + b + 1/2} \leq 1, 
	\end{align*}
	so we can set $h_n(x) = 1$ when $x = 0$. Therefore, $\mathbb{E}(\kappa | x)$ is bounded above by the function $h_n(x)$ in (\ref{hnx}) for all $x \in \mathbb{R}$. 
	
	Now, observe from (\ref{h1x}) that for fixed $n$, $h_1(x)$ is strictly decreasing in $|x|$. Therefore, 
	\begin{equation*}
	\displaystyle \sup_{|x| > \sqrt{\rho \log \left( \frac{1}{a_n} \right)}} h_1(x) \leq C_{n, \eta} \left[ \rho \log \left( \frac{1}{a_n} \right) \displaystyle \int_{0}^{\rho \log \left( \frac{1}{a_n} \right)} t^{b-1/2} e^{-t/2} dt \right]^{-1},
	\end{equation*}
	for any fixed $n$ and $\rho > 0$. Since $a_n \rightarrow 0$ as $n \rightarrow \infty$, this implies that
	\begin{equation} \label{limsuph1x}
	\displaystyle \lim_{n \rightarrow \infty} \displaystyle \sup_{|x| > \sqrt{\rho \log \left( \frac{1}{a_n} \right)}} h_1(x) = 0.
	\end{equation}
	Letting $K \equiv K(b, \eta, \delta) =  \left( b + \frac{3}{2} \right) / (\eta \delta)^{b+3/2}$, we have from (\ref{h2x}) and the fact that $0 < a_n < 1$ for all $n$ and $a_n \rightarrow 0$ as $n \rightarrow 0$ that
	
	\begin{align*}
	\displaystyle \lim_{n \rightarrow \infty} h_2 \left( \sqrt{\rho \log \left( \frac{1}{a_n} \right)} \right) & = K \displaystyle \lim_{n \rightarrow \infty} \frac{(1 - \eta)^{a_n}}{a_n} \sqrt{ \rho \log \left( \frac{1}{a_n} \right) } e^{-\frac{\eta(1-\delta)}{2} \rho \log \left( \frac{1}{a_n} \right)} 
	\end{align*}
	\begin{align*}
	& \leq K \sqrt{\rho} \displaystyle \lim_{n \rightarrow \infty} \frac{1}{a_n} \sqrt{ \log \left( \frac{1}{a_n} \right)} e^{-\frac{\eta (1 - \delta)}{2} \log(a_n^{-\rho})}  \\
	& = K \sqrt{\rho} \displaystyle \lim_{n \rightarrow \infty} \sqrt{ \log \left( \frac{1}{a_n} \right) } (a_n)^{\frac{\eta(1-\delta)}{2} \left( \rho - \frac{2}{\eta (1-\delta)} \right)} \\
	& = \left\{ \begin{array}{ll}0 & \textrm{if } \rho > \frac{2}{\eta(1-\delta)}, \\ \infty & \textrm{otherwise,} \end{array} \right. \end{align*}
	from which it follows that
	\begin{equation} \label{limsuph2x}
	\displaystyle \lim_{n \rightarrow \infty} \displaystyle \sup_{|x| > \sqrt{\rho \log \left( \frac{1}{a_n} \right)}} h_2(x) = \left\{ \begin{array}{ll}0 & \textrm{if } \rho > \frac{2}{\eta(1-\delta)}, \\ \infty & \textrm{otherwise.} \end{array} \right. 
	\end{equation} 
	Combining (\ref{limsuph1x}) and (\ref{limsuph2x}), it is clear that
	\begin{equation*}
	\displaystyle \lim_{n \rightarrow \infty} \displaystyle \sup_{|x| > \sqrt{\rho \log \left( \frac{1}{a_n} \right)}} h_n(x) = \left\{ \begin{array}{ll}0 & \textrm{if } \rho > \frac{2}{\eta(1-\delta)}, \\ \infty & \textrm{otherwise,} \end{array} \right. 
	\end{equation*}
	that is, $h_n(x)$ satisfies (\ref{hnxlimsup}).
\end{proof}
\begin{proof}[Proof of Theorem \ref{Theorem:6}]
	Fix $\eta \in (0,1)$ and $\delta \in (0, 1)$, and choose any $\rho > \frac{2}{\eta(1-\delta)}$. By Lemma \ref{Lemma:1}, we have that the event $\{ \mathbb{E}( \kappa_i | X_i ) \geq 0.5 \}$ implies $\{ h_n(X_i) \geq 0.5 \}$, where $h_n(x)$ is as defined in (\ref{hnx}). Therefore,
	\begin{align*} \label{TypeIIbound1}
	t_{2i} & = \Pr [ \mathbb{E} (\kappa_i | X_i) \geq 0.5 \big| H_{1i} \textrm{ is true} ] \\
	& \leq \Pr ( h_n(X_i) \geq 0.5 \big| H_{1i} \textrm{ is true} ) \\
	& = \Pr \left( h_n(X_i) \geq 0.5, |X_i| > \sqrt{ \rho \log \left( \frac{1}{a_n} \right)} \hspace{.1cm} \bigg| H_{1i} \textrm{ is true} \right) + \\
	& \qquad \Pr \left( h_n(X_i) \geq 0.5, |X_i| \leq \sqrt{ \rho \log \left( \frac{1}{a_n} \right)} \hspace{.1cm} \bigg| H_{1i} \textrm{ is true} \right) \\
	& \leq \Pr \left( h_n(X_i) \geq 0.5 \bigg| |X_i| > \sqrt{ \rho \log \left( \frac{1}{a_n} \right)}, H_{1i} \textrm{ is true} \right) + \\
	& \qquad \Pr \left( |X_i| \leq \sqrt{ \rho \log \left( \frac{1}{a_n} \right)} \hspace{.1cm} \bigg| H_{1i} \textrm{ is true} \right) \numbereqn
	\end{align*}
	We will consider the two terms in (\ref{TypeIIbound1}) separately. Recall that $h_n(x)$ from (\ref{hnx}) is a measurable and nonnegative. We also see that (\ref{hnx}) is decreasing in $|x|$, and thus,
	
	\begin{equation*}
	\mathbb{E} \left( h_n(X_i) \bigg| |X_i| > \sqrt{ \rho \log \left( \frac{1}{a_n} \right)}, H_{1i} \textrm{ is true} \right) 
	\end{equation*}
	is well-defined and bounded for sufficiently large $n$. By Markov's inequality, we have for sufficiently large $n$,
	\begin{align*}
	& \Pr \left( h_n(X_i) \geq 0.5 \bigg| |X_i| > \sqrt{ \rho \log \left( \frac{1}{a_n} \right)}, H_{1i} \textrm{ is true} \right) \\
	& \qquad \leq 2 \mathbb{E} \left( h_n(X_i) \bigg| |X_i| > \sqrt{ \rho \log \left( \frac{1}{a_n} \right)}, H_{1i} \textrm{ is true} \right) \\
	& \qquad \leq 2 \left( \displaystyle \sup_{|X_i| > \sqrt{\rho \log \left( \frac{1}{a_n} \right)}} h_n(X_i) \right),
	\end{align*}
	from which it follows, by Lemma \ref{Lemma:1}, that
	\begin{equation} \label{TypeIIbound2}
	\displaystyle \lim_{n \rightarrow \infty} \Pr \left( h_n(X_i) \geq 0.5 \bigg| |X_i| > \sqrt{ \rho \log \left( \frac{1}{a_n} \right)}, H_{1i} \textrm{ is true} \right) = 0.
	\end{equation}
	By assumption, $\lim_{n \rightarrow \infty} \frac{a_n}{p_n} \in (0, \infty)$. Thus, by the third and fourth conditions of Assumption \ref{Assumption:1}, we have $\lim_{n \rightarrow \infty} \log(\frac{1}{a_n}) / \psi_n^2 = C/2$. To see this, note that $\frac{1-p_n}{p_n} \sim \frac{1}{p_n}$. Combining this with the third and fourth conditions implies that $\frac{2 \log(1/p_n)}{\psi_n^2} \rightarrow C$, and then we use our assumption that $a_n/p_n \rightarrow d, d > 0$. Thus, for all sufficiently large $n$,
	\begin{align*} \label{TypeIIbound3}
	& \Pr \left( |X_i| \leq \sqrt{ \rho \log \left( \frac{1}{a_n} \right)} \hspace{.1cm} \bigg| H_{1i} \textrm{ is true} \right) = \Pr \left( |Z| \leq \sqrt{\rho} \sqrt{ \frac{\log \left( \frac{1}{a_n} \right) }{1+\psi_n^2} }\right) \\
	& \qquad = \Pr \left( |Z| \leq \sqrt{\rho} \sqrt{ \frac{\log \left( \frac{1}{a_n} \right) }{\psi_n^2} } (1+o(1)) \right) \textrm{ as } n \rightarrow \infty \\
	& \qquad = \Pr \left( |Z| \leq \sqrt{ \frac{\rho C}{2}} ( 1+o(1)) \right) \textrm{ as } n \rightarrow \infty \\
	& \qquad = \left[ 2 \Phi \left(\sqrt{ \frac{\rho C}{2}}\right) - 1 \right] \left(1+o(1) \right) \textrm{ as } n \rightarrow \infty. \numbereqn
	\end{align*}
	Combining (\ref{TypeIIbound1})-(\ref{TypeIIbound3}), we thus have
	\begin{equation*}
	t_{2i} \leq \left[ 2 \Phi \left(\sqrt{ \frac{\rho C}{2}}\right) - 1 \right](1+o(1)),
	\end{equation*}
	as $n \rightarrow \infty$.
\end{proof}

\begin{proof}[Proof of Theorem \ref{Theorem:7}]
	By definition, the probability of a Type II error for the $i$th decision is given by 
	\begin{equation*}
	t_{2i} = P \left( \mathbb{E}(1 - \kappa_i) \leq \frac{1}{2} \bigg| H_{1i} \textrm{ is true } \right). 
	\end{equation*}
	For any $n$, we have by Theorem \ref{Th:1} that
	\begin{equation*}
	\left\{ e^{X_i^2/2} \left( \frac{a_n }{a_n + b + 1/2}  \right) \leq \frac{1}{2} \right\} \subseteq \left\{ \mathbb{E}(1 - \kappa_i | X_i) \leq \frac{1}{2} \right\}.
	\end{equation*}
	Therefore,
	\begin{align*}
	t_{2i} & = \Pr \left( \mathbb{E} (1 - \kappa_i | X_i) \leq \frac{1}{2} \bigg| H_{1i} \textrm{ is true} \right) \\
	& \geq \Pr \left( e^{X_i^2/2} \left( \frac{a_n }{a_n + b + 1/2}  \right) \leq \frac{1}{2} \bigg| H_{1i} \textrm{ is true} \right)  \\
	& = \Pr \left( X_i^2 \leq 2 \log \left( \frac{a_n + b + 1/2}{2 a_n }  \right) \bigg|  H_{1i} \textrm{ is true} \right). \numbereqn \label{TypeIIlowerbound1}
	\end{align*}
	Since $X_i \sim N(0, 1 + \psi^2)$ under $H_{1i}$, we have by the second condition in Assumption \ref{Assumption:1} that $\displaystyle \lim_{n \rightarrow \infty} \frac{\psi_n^2}{1 + \psi_n^2} \rightarrow 1$. From (\ref{TypeIIlowerbound1}), we have for sufficiently large $n$,
	
	\begin{align*}
	t_{2i} & \geq \Pr \left( |Z| \leq \sqrt{ \frac{2 \log \left( \frac{a_n + b + 1/2 }{2 a_n}  \right)}{\psi^2}} (1 + o(1)) \right) \textrm{ as } n \rightarrow \infty \\
	& \geq \Pr \left( |Z| \leq \sqrt{ \frac{ \log \left( \frac{1}{2 a_n}  \right)}{\psi^2}} (1 + o(1)) \right) \textrm{ as } n \rightarrow \infty \\
	& = \Pr ( |Z| \leq \sqrt{C} ) (1 + o(1)) \textrm{ as } n \rightarrow \infty \\
	& = 2 [ \Phi (\sqrt{C}) - 1] (1 + o(1)) \textrm{ as } n \rightarrow \infty, 
	\end{align*}
	where we used the assumption that $\lim_{n \rightarrow \infty} \frac{a_n}{p_n} \in (0, \infty)$ and Assumption \ref{Assumption:1}.
\end{proof}

\begin{proof}[Proof of Theorem \ref{Th:8}]
	Fix $\eta \in (0,1), \delta \in (0, 1),$ and $\xi \in (0, 1/2)$, and choose $\rho > \frac{2}{\eta (1 - \delta)}$. Since the $\kappa_i$'s, $i = 1, ..., n$ are \textit{a posteriori} independent, the Type I and Type II error probabilities $t_{1i}$ and $t_{2i}$ are the same for every test $i, i = 1, ..., n$. By Theorems \ref{Theorem:4} and \ref{Theorem:5}, we have for large enough $n$,
	
	\begin{align*}
	& 2 \left(1 - \Phi \left( \sqrt{ \frac{2}{\xi (1 - \delta)} \left[ \log \left( \frac{\left(b + \frac{1}{2} \right) \left( 1 - \xi \right)^{a_n} }{a_n (\xi \delta)^{b + \frac{1}{2}} \left( \frac{1}{2} - \xi \right) } \right)  \right]  } \right) \right) \leq t_{1i} \\
	& \leq \frac{2 \sqrt{2} a_n}{\sqrt{\pi} (a_n + b + 1/2)} \left[ \log \left( \frac{a_n + b + 1/2}{2 a_n } \right) \right]^{-1/2}.
	\end{align*}
	Taking the limit as $n \rightarrow \infty$ of all the terms above, we have 
	\begin{equation} \label{typeIlim}
	\displaystyle \lim_{n \rightarrow \infty} t_{1i} = 0
	\end{equation}
	for the $i$th test, under the assumptions on the hyperparameters $a_n$ and $b$.
	
	By Theorems \ref{Theorem:4} and \ref{Theorem:5}, we also have 
	\begin{equation} \label{typeIIlim}
	\left[ 2 \Phi(\sqrt{C}) - 1 \right](1+ o(1)) \leq t_{2i} \leq \left[ 2 \Phi \left( \sqrt{\frac{ \rho C}{2}} \right) - 1 \right] (1 + o(1)).
	\end{equation}
	Therefore, we have by (\ref{typeIlim}) and (\ref{typeIIlim}) that as $n \rightarrow \infty$, the asymptotic risk (\ref{twogroupsrisk}) of the classification rule (\ref{thresholdingrule}), $R_{NBP}$, can be bounded as follows:
	\begin{equation} \label{NBPriskbounds}
	np(2 \Phi (\sqrt{C}) - 1)(1+ o(1)) \leq R_{NBP} \leq np \left[ 2 \Phi \left( \sqrt{ \frac{\rho C}{ 2 }} \right) - 1 \right](1+o(1).
	\end{equation}
	Therefore, from (16) and (\ref{NBPriskbounds}), we have as $n \rightarrow \infty$,
	\begin{equation} \label{riskratiobounds}
	1 \leq \lim \inf_{n \rightarrow \infty} \frac{R_{NBP}}{R_{Opt}^{BO}} \leq \lim \sup_{n \rightarrow \infty} \frac{R_{NBP}}{R_{Opt}^{BO}} \leq \frac{2 \Phi \left( \sqrt{\frac{ \rho C }{2}} \right) - 1}{2 \Phi (\sqrt{C}) - 1}.
	\end{equation} 
	The testing rule (18) does not depend on how $\eta \in (0,1), \delta \in (0,1)$ and $\rho > 2/(\eta ( 1-\delta))$ are chosen, and thus, the ratio $R_{NBP}/R_{Opt}^{BO}$ is also free of these constants. By continuity of $\Phi$, we can take the infimum over all $\rho$'s in the rightmost term in (\ref{riskratiobounds}), and the inequalities remain valid. The infimum of $\rho$ is obviously 2, and so from (\ref{riskratiobounds}), we have
	\begin{equation} \label{riskratiobounds2}
	1 \leq \lim \inf_{n \rightarrow \infty} \frac{R_{NBP}}{R_{Opt}^{BO}} \leq \lim \sup_{n \rightarrow \infty} \frac{R_{NBP}}{R_{Opt}^{BO}} \leq \frac{2 \Phi ( \sqrt{C}) - 1}{2 \Phi (\sqrt{C}) - 1}.
	\end{equation} 
	We clearly see from (\ref{riskratiobounds2}) that classification rule (\ref{thresholdingrule}) under the $\textrm{NBP}_n$ prior (10) is ABOS, i.e.
	\begin{equation*}
	\frac{ R_{NBP}}{R_{Opt}^{BO}} \rightarrow 1 \textrm{ as } n \rightarrow \infty.
	\end{equation*}
\end{proof}

\section{Proofs for Section 3.4} \label{App:C}
Our proofs in this section follow from the proof of Theorem 10 of \cite{GhoshTangGhoshChakrabarti2016}, as well as Theorems \ref{Th:2} through Theorem \ref{Th:3} established in this paper.

\begin{proof}[Proof of Theorem \ref{Theorem:9}]
	Under thresholding rule (\ref{thresholdingruleEB}), the probability of a Type I error for the $i$th decision is given by
	\begin{align*}\label{EBtypeIerror}
	& \widetilde{t}_{1i} = \Pr \left( \mathbb{E}(1 - \kappa_i | X_i, \widehat{a}_n^{ES} ) > \frac{1}{2} \bigg| H_{0i} \textrm{ is true} \right) \\
	& \qquad = \Pr \left( \mathbb{E}(1 - \kappa_i | X_i, \widehat{a}_n^{ES} ) > \frac{1}{2}, \widehat{a}_n^{ES} \leq 2 \alpha_n \bigg| H_{0i} \textrm{ is true} \right)  \\
	& \qquad \qquad + \Pr \left( \mathbb{E}(1 - \kappa_i | X_i, \widehat{a}_n^{ES} ) > \frac{1}{2}, \widehat{a}_n^{ES} > 2 \alpha_n \bigg| H_{0i} \textrm{ is true} \right), \numbereqn
	\end{align*}
	where $\alpha_n$ is defined in (\ref{alphan}). To obtain an upper bound on $\widetilde{t}_{1i}$, we consider the two terms in (\ref{EBtypeIerror}) separately. By Theorem 2.1, we see that $\mathbb{E}(1 - \kappa_i | X_i)$ is nondecreasing in $a_n$. Thus, $\mathbb{E}(1 - \kappa_i | X_i, \widehat{a}_n^{ES} ) \leq \mathbb{E}(1 - \kappa_i | X_i, 2 \alpha_n)$ whenever $\widehat{a}_n^{ES} \leq 2 \alpha_n$. We have
	\begin{align*} \label{EBtypeIupperbound1}
	& \Pr \left( \mathbb{E}(1 - \kappa_i | X_i, \widehat{a}_n^{ES} ) > \frac{1}{2}, \widehat{a}_n^{ES} \leq 2 \alpha_n \bigg| H_{0i} \textrm{ is true} \right) \\
	& \qquad \leq \Pr \left( \mathbb{E} (1-\kappa_i | X_i, 2 \alpha_n ) > \frac{1}{2} \bigg| H_{0i} \textrm{ is true} \right) \\
	& \qquad \leq \frac{4 \alpha_n}{\sqrt{\pi} (  2 \alpha_n + b + 1/2)} \left[ \log \left( \frac{ 2 \alpha_n + b + 1/2}{4 \alpha_n} \right) \right]^{-1/2} (1 + o(1)). \numbereqn
	\end{align*}
	For the second term in (\ref{EBtypeIerror}), we have
	\begin{align*} \label{EBtypeIupperbound2}
	& \Pr \left( \mathbb{E}(1 - \kappa_i | X_i, \widehat{a}_n^{ES} ) > \frac{1}{2}, \widehat{a}_n^{ES} > 2 \alpha_n \bigg| H_{0i} \textrm{ is true} \right) \\
	& \qquad \leq \Pr ( \widehat{a}_n^{ES} > 2 \alpha_n | H_{0i} \textrm{ is true} ) \\
	& \qquad \leq \frac{1/\sqrt{\pi}}{n^{c_1/2} \sqrt{\log n}} + e^{-(2 \log 2 - 1) \alpha_n (1+o(1))}, \numbereqn
	\end{align*}
	where the last inequality follows from the proof of Theorem 10 in \cite{GhoshTangGhoshChakrabarti2016}. Thus, since $\alpha_n \sim 2 \beta p_n$ by (\ref{alpharelation}), we combine (\ref{EBtypeIupperbound1}) and (\ref{EBtypeIupperbound2}) to obtain an upper bound on $\widetilde{t}_{1i}$,
	\begin{align*}
	& \widetilde{t}_{1i} \leq \frac{4 \alpha_n}{\sqrt{\pi} ( 2 \alpha_n + b + 1/2)} \left[ \log \left( \frac{ 2 \alpha_n + b + 1/2}{4 \alpha_n} \right) \right]^{-1/2} (1 + o(1))  \\
	& \qquad \qquad + \frac{1 / \sqrt{\pi}}{n^{c_1/2} \sqrt{\log n}} + e^{-2 (2 \log 2 - 1) \beta np_n(1+o(1))}.
	\end{align*}
	To obtain the lower bound, note that by (\ref{EBtypeIerror}), we immediately have
	\begin{equation} \label{EBtypeIlowerbound}
	\widetilde{t}_{1i} \geq \Pr \left( \mathbb{E}(1- \kappa_i | X_i, \widehat{a}_n^{ES}) > \frac{1}{2}, \widehat{a}_n^{ES} \leq 2 \alpha_n \bigg| H_{0i} \textrm{ is true} \right).
	\end{equation}
	By the proof for Theorem \ref{Theorem:5}, we have that for fixed $\xi \in (0, 1/2)$ and $\delta \in (0, 1)$, 
	\begin{equation} \label{upperboundEkappa} 
	\mathbb{E} (\kappa_i | X_i) \leq \xi + \frac{ \left( b + \frac{1}{2} \right) (1 - \xi)^{a_n} }{ a_n ( \xi \delta)^{b + \frac{1}{2}} } \exp \left( - \frac{\xi (1 - \delta)}{2} X_i^2 \right).
	\end{equation}
	The right-hand side of (\ref{upperboundEkappa}) is a nonincreasing function in $a_n$. Thus, whenever $\widehat{a}_n^{ES} \leq 2 \alpha_n$, we have
	\begin{align*}
	\left\{ \mathbb{E}(1 -\kappa_i | X_i, \widehat{a}_n ) > \frac{1}{2}, \widehat{a}_n^{ES} \leq 2 \alpha_n \right\} \supseteq \left\{ \frac{\left(b + \frac{1}{2} \right) (1 - \xi)^{2 \alpha_n}}{2 \alpha_n (\xi \delta)^{b+1/2}} \exp \left( - \frac{\xi(1-\delta)}{2} X_i^2 \right) < \frac{1}{2} - \xi \right\},
	\end{align*}
	from which, by Theorem 3.2 and (\ref{EBtypeIlowerbound}), we automatically attain the lower bound,
	\begin{equation*}
	\widetilde{t}_{1i} \geq 1 - \Phi \left( \sqrt{ \frac{2}{\xi (1-\delta)} \left[ \log \left( \frac{ \left(b+\frac{1}{2} \right) (1 - \xi)^{2 \alpha_n}}{2 \alpha_n ( \xi \delta)^{b+1/2}} \right) \right] } \right) ( 1 + o(1)) \textrm{ as } n \rightarrow \infty.
	\end{equation*}
\end{proof}
\begin{proof}[Proof of Theorem \ref{Theorem:10}]
	Fix $\gamma \in  (0, 1/c_2)$. Decompose the probability of a Type II error under (\ref{thresholdingruleEB}) as 
	\begin{align*} \label{EBtypeIIerror}
	& \widetilde{t}_{2i} = \Pr \left( \mathbb{E}(\kappa_i | X_i, \widehat{a}_n^{ES} ) \geq \frac{1}{2} \bigg| H_{1i} \textrm{ is true} \right) 
	\end{align*}
	\begin{align*}
	& =  \Pr \left( \mathbb{E}(\kappa_i | X_i, \widehat{a}_n^{ES} ) \geq \frac{1}{2}, \widehat{a}_n^{ES} \leq \gamma \alpha_n \bigg| H_{1i} \textrm{ is true} \right) + \\ 
	& \qquad \qquad \Pr \left( \mathbb{E}(\kappa_i | X_i, \widehat{a}_n^{ES} ) \geq \frac{1}{2}, \widehat{a}_n^{ES} > \gamma \alpha_n \bigg| H_{1i} \textrm{ is true} \right). 
	\numbereqn
	\end{align*}
	
	To obtain an upper bound on $\widetilde{t}_{2i}$, we consider the two terms in (\ref{EBtypeIIerror}) separately. For the first term in (\ref{EBtypeIIerror}), we have
	\begin{align*} \label{EBtypeIIupperbound1}
	& \Pr \left( \mathbb{E}(\kappa_i | X_i, \widehat{a}_n^{ES} ) \geq \frac{1}{2}, \widehat{a}_n^{ES} \leq \gamma \alpha_n \bigg| H_{1i} \textrm{ is true} \right) \\
	& \qquad \leq \Pr ( \widehat{a}_n^{ES} \leq \gamma \alpha_n | H_{1i} \textrm{ is true} ) \\
	& \qquad \leq \frac{(1- c_2 \gamma)^{-2} (1- \alpha_n)}{n \alpha_n} ( 1+ o(1))  \\
	& \qquad \rightarrow 0 \textrm{ as } n \rightarrow \infty, \numbereqn
	\end{align*}
	where the last two steps follow from the proof of Theorem 11 in \cite{GhoshTangGhoshChakrabarti2016}. 
	
	We now focus on bounding the second term in (\ref{EBtypeIIerror}). By Theorem \ref{Th:1}, $\mathbb{E}(1-\kappa_i | X_i)$ is nondecreasing in $a_n$, and so $\mathbb{E}(\kappa_i | X_i)$ is nonincreasing  in $a_n$. Thus, for sufficiently large $n$, we have  $\mathbb{E}(\kappa_i | X_i, \widehat{a}_n^{ES} ) \leq \mathbb{E} ( \kappa_i | X_i, \gamma \alpha_n )$ for $\widehat{a}_n^{ES} > \gamma \alpha_n$ and that  
	\begin{equation*}
	\left\{ \mathbb{E} (\kappa_i | X_i, \gamma \alpha_n) \geq 0.5 | H_{1i} \textrm{ is true} \right\} \subseteq \left\{ h_n(X_i, \gamma \alpha_n) \geq 0.5 | H_{1i} \textrm{ is true} \right\},
	\end{equation*}
	where $h_n(X_i, \gamma \alpha_n)$ denotes that we substitute $a_n$ with $\gamma \alpha_n$ in (\ref{hnx}). Using the same arguments as in the proof of Theorem \ref{Theorem:6}, along with the fact that $\alpha_n \sim 2 \beta p_n$ (by (\ref{alpharelation})), we obtain as an upper bound for the second term in (\ref{EBtypeIIerror}),
	\begin{align*} \label{EBtypeIIupperbound2}
	& \Pr \left( \mathbb{E}(\kappa_i | X_i, \widehat{a}_n^{ES} ) \geq \frac{1}{2}, \widehat{a}_n^{ES} > \gamma \alpha_n \bigg| H_{1i} \textrm{ is true}  \right) \\
	& \qquad \leq \Pr\left( \mathbb{E}(\kappa_i | X_i, \gamma \alpha_n) \geq \frac{1}{2} | H_{1i} \textrm{ is true} \right) \\
	& \qquad \leq \left[2 \Phi \left( \sqrt{\frac{\rho C}{2}} \right) - 1 \right](1+ o(1)) \textrm{ as } n \rightarrow \infty.  \numbereqn
	\end{align*}
	From (\ref{EBtypeIIerror})-(\ref{EBtypeIIupperbound2}), an upper bound on the probability of Type II error under (\ref{thresholdingruleEB}) is
	\begin{equation*}
	\widetilde{t}_{2i} \leq \left[ 2 \Phi \left( \sqrt{\frac{\rho C}{2}} \right) - 1 \right] (1+o(1)) \textrm{ as } n \rightarrow \infty.
	\end{equation*}
	To obtain a lower bound on $\widetilde{t}_{2i}$, we note that by (\ref{EBtypeIIerror}), 
	\begin{align*} \label{EBtypeIIlower1}
	& \widetilde{t}_{2i} \geq \Pr \left( \mathbb{E}(\kappa_i | X_i, \widehat{a}_n^{ES}) \geq \frac{1}{2}, \widehat{a}_n^{ES} > \gamma \alpha_n \bigg| H_{1i} \textrm{ is true} \right)  \\
	& \geq \Pr \left( \mathbb{E}( \kappa_i | X_i, \widehat{a}_n^{ES} ) \geq \frac{1}{2} \right) - \Pr( \widehat{a}_n^{ES} \leq \gamma \alpha_n)  \\
	& \rightarrow 2 \left[  \Phi(\sqrt{C}) - 1 \right] (1+o(1)) - o(1),
	\end{align*}
	where we use the result in Theorem \ref{Theorem:7}, the fact that $\mathbb{E}(\kappa_i | X_i )$ is nondecreasing in $a_n$, and the fact that $\Pr( \widehat{a}_n^{ES} \leq \gamma \alpha_n)$ is asymptotically vanishing (by (\ref{EBtypeIIupperbound1})) to arrive at the final inequality.
\end{proof}

\section{Sampling from the NBP Model} \label{App:D}
\subsection{No Prior on the Hyperparameter $a$}
Suppose that there is no prior placed on the hyperparameter $a$. By the reparametrization of $\sigma_i^2 = \lambda_i \xi_i, i = 1, \ldots, n,$  given in (\ref{NBPhiersingle}) and letting $\kappa_i = 1/(1+ \lambda_i \xi_i)$, the full conditional distributions for (\ref{NBPhier}) are 
\begin{equation} \label{nbpconditionalposteriors}
\begin{array}{ccl}
\theta_i \hspace{.1cm} \big| \hspace{.1cm} \textrm{rest} & \sim & \mathcal{N} \big( (1-\kappa_i)X_i, 1-\kappa_i \big), i = 1, ..., n, \\
\lambda_i \hspace{.1cm} \big| \hspace{.1cm} \textrm{rest} & \sim & \mathcal{GIG} \left( \frac{\theta_i^2}{\xi_i}, 2, a - \frac{1}{2} \right), i = 1, ..., n, \\
\xi_i \hspace{.1cm} \big| \hspace{.1cm} \textrm{rest} & \sim & \mathcal{IG} \left( b+\frac{1}{2}, \frac{\theta_i^2}{2 \lambda_i} + 1 \right), i = 1, ..., n,
\end{array}
\end{equation}
where $\mathcal{GIG}(c, d, p)$ denotes a generalized inverse Gaussian (giG) density with $f(x; c, d, p) \propto x^{(p-1)} e^{-(c/x+dx)/2}$. Therefore, the NBP model (\ref{NBPhier}) -- and consequently, thresholding rules (\ref{thresholdingrule}) and (\ref{thresholdingruleEB}) -- can be implemented straightforwardly with Gibbs sampling utilizing the full conditionals in (\ref{nbpconditionalposteriors}). Moreover, since the full conditionals are independent, we can update the $\theta_i$'s, $\lambda_i$'s, and $\xi_i$'s efficiently using block updates.

\subsection{Uniform Prior on the Hyperparameter $a$}
In the case that a prior is placed on $a$, the steps for sampling from the full conditionals for $(\theta_i, \lambda_i, \xi_i), i = 1, \ldots, n$ from (\ref{nbpconditionalposteriors}) remain the same. However, we now also need to sample from the full conditional of $a$. When $a \sim \mathcal{U}(1/n, 1)$, the full conditional for $a$ is proportional to
\begin{equation} \label{conditionalaunif}
\pi(a | \textrm{rest} ) \propto \left( \frac{\Gamma(a+b)}{\Gamma(a)} \right)^{n} \left( \displaystyle \prod_{i=1}^{n} (\sigma_i^2)^{a-1} (1+\sigma_i^2)^{-a-b} \right) \mathbb{I} \{ 1/n \leq a \leq 1 \},
\end{equation}
where $\sigma_i^2 = \lambda_i \xi_i$. Using (\ref{conditionalaunif}), we update $a$ using a Metropolis-Hastings random walk. For our proposal distribution, we use a truncated normal density on the interval $[1/n, 1]$. If $a$ is the current value of the chain, a new value $a^*$ will be generated from the proposal distribution,
\begin{equation} \label{truncatednormaldist}
q(a^* | a) =\frac{\phi \left( \frac{a^* - a}{\omega} \right)}{\omega \left( \Phi \left( \frac{1-a}{\omega} \right) -\Phi \left( \frac{1/n-a}{\omega} \right) \right)} \mathbb{I} \{ 1/n \leq a^* \leq 1 \},
\end{equation}
where $\phi(\cdot)$ and $\Phi(\cdot)$ denote the standard normal probability density function (pdf) and cumulative distribution function (cdf) respectively, and $\omega > 0$ is a scaling parameter that is properly calibrated to control the Metropolis-Hastings acceptance rate.  Given a candidate state $a^*$ drawn from $q(a^* | a)$, it then follows from (\ref{conditionalaunif}) and (\ref{truncatednormaldist}) that $a^*$ is accepted with probability,

\begin{equation*}
\min \left\{ 1, \left( \frac{\Gamma (a^*+b) \Gamma(a)}{\Gamma(a+b) \Gamma (a^*)} \right)^{n} \displaystyle \left( \prod_{i=1}^{n} \left( \frac{\sigma_i^2}{1+\sigma_i^2} \right)^{a^*-a} \right) \left[ \frac{\Phi \left( \frac{1-a}{\omega} \right) -\Phi \left( \frac{1/n-a}{\omega} \right)}{\Phi \left( \frac{1-a^*}{\omega} \right) -\Phi \left( \frac{1/n-a^*}{\omega} \right)} \right] \right\},
\end{equation*}
where $\sigma_i^2, i=1, \ldots, n,$ is taken as the product of the $\lambda_i$ and  $\xi_i$ from the most recent Gibbs sampling updates for $(\lambda_i, \xi_i), i = 1, \ldots, n$.  We tune $\omega$ so that the acceptance rate is between 20 and 40 percent.

\subsection{Truncated Cauchy Prior on the Hyperparameter $a$}
If we place a truncated Cauchy prior on $a$ where $a \in [1/n, 1]$, i.e. $\pi(a) = [\textrm{arctan}(1) -\textrm{arctan}(1/n)]^{-1} (1+a)^{-1} \mathbb{I} \{ 1/n < a < 1 \}$, the full conditional for $a$ is proportional to
\begin{equation} \label{conditionalatruncatedcauchy}
\pi(a | \textrm{rest} ) \propto \left( \frac{\Gamma(a+b)}{ (1+a) \Gamma(a)} \right)^{n}  \left( \displaystyle \prod_{i=1}^{n} (\sigma_i^2)^{a-1} (1+\sigma_i^2)^{-a-b} \right)  \mathbb{I} \{ 1/n \leq a \leq 1 \}.
\end{equation}
As before, we use Metropolis-Hastings to update $a$. We use the truncated normal density $q(a^* |a )$ from (\ref{truncatednormaldist}) as the proposal distribution. Given a candidate state $a^*$ drawn from $q(a^* | a)$ in (\ref{truncatednormaldist}), it follows from (\ref{conditionalatruncatedcauchy}) that $a^*$ is accepted with probability,
\begin{equation*}
\min \left\{ 1, \left( \frac{(1+a) \Gamma (a^*+b) \Gamma(a)}{(1+a^*) \Gamma(a+b) \Gamma (a^*)} \right)^{n} \displaystyle \left( \prod_{i=1}^{n} \left( \frac{\sigma_i^2}{1+\sigma_i^2} \right)^{a^*-a} \right) \left[ \frac{\Phi \left( \frac{1-a}{\omega} \right) -\Phi \left( \frac{1/n-a}{\omega} \right)}{\Phi \left( \frac{1-a^*}{\omega} \right) -\Phi \left( \frac{1/n-a^*}{\omega} \right)} \right]\right\},
\end{equation*}
where $\sigma_i^2, i=1, \ldots, n,$ is taken as the product of the $\lambda_i$ and  $\xi_i$ from the most recent Gibbs sampling updates for $(\lambda_i, \xi_i), i = 1, \ldots, n$.  We tune $\omega$ so that the acceptance rate is between 20 and 40 percent.

\subsection{Convergence of the MCMC Algorithm}
To assess the convergence and the mixing of the MCMC algorithms for the hierachical Bayes approaches described in Sections D.2 and D.3, we consider two chains with different starting values: 1) $\theta^{(0)}_i = -15, i= 1, \ldots, n,$ and 2) $\theta^{(0)}_i = 15, i= 1, \ldots, n$. In our simulation studies, the true $\bm{\theta}_0$ was generated from
\begin{align*}
\theta_{0i} \overset{iid}{\sim} (1-p) \delta_{0} + p \mathcal{N}(0, \psi^2), i =1, \ldots, n,
\end{align*}
with $\psi = \sqrt{2 \log (500) } = 3.53$. Thus, these initial values for $\theta_{i}^{(0)}, i = 1, \ldots, n$, are all far away from a `typical' value of $\theta_{0i}$. We found that in both cases, the MCMC algorithms still converged very rapidly (usually within 100 iterations), giving very similar posterior estimates for $\bm{\theta}$ after discarding the first 5000 iterations as burnin.

To illustrate this, we plot in Figure \ref{MCMCplots} the history plots for one noncoefficient coefficient ($\theta_{0i} = 7.225$) and one null coefficient $(\theta_{0i} = 0$) when the sparsity level is $p=0.2$. For the nonnull coefficient, we see that the chains mix well and rapidly converge to a stationary distribution centered around the true value of $\theta_{0i}$. For the null coefficient, the chains rapidly converge to a stationary distribution centered around zero.

\begin{figure}[h!]
	\centering
	\includegraphics[width=.75\textwidth]{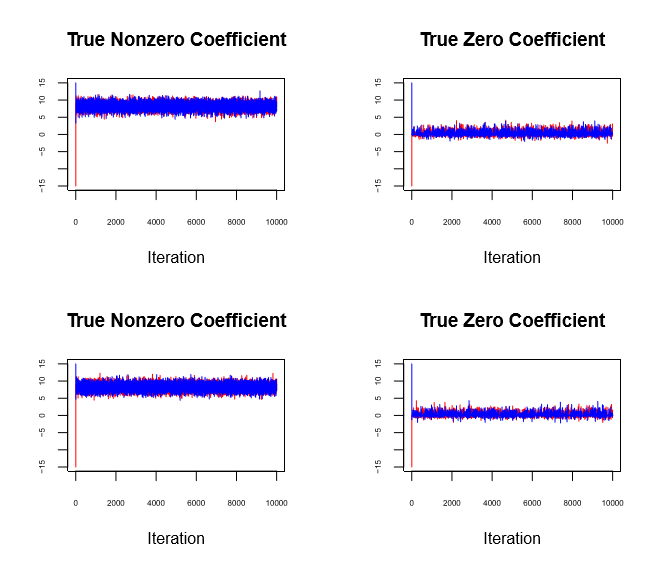} 
	
	\caption{History plots of the 10,000 draws from the MCMC algorithm for the NBP-UNIF (top panel) and NBP-TC models (bottom panel) for a single $\theta_{0i}$. The plots on the left are for a $\theta_{0i}$ whose true value is equal to 7.225, and the plots on the right are for a $\theta_{0i}$ whose true value is equal to 0.}
	\label{MCMCplots}
\end{figure}


\begin{thebibliography}{99}

\bibitem{ArmaganClydeDunson2011}%1
Armagan~A, Clyde~M, Dunson~DB. Generalized beta mixtures of gaussians. NeurIPS 2011; 24:523-531.

\bibitem{ArmaganDunsonLee2013}%2
Armagan~A, Dunson~DB, Lee~J. Generalized double pareto shrinkage. Statist. Sinica. 2013;23:119-143.

\bibitem{BenjaminiHochberg1995}%3
Benjamini~Y, Hochberg~Y. Controlling the false discovery rate: A practical and powerful approach to multiple testing. J. R. Stat. Soc. Ser. B Stat. Methodol. 1995;57:289-300.

\bibitem{Berger1980}%4
Berger~J. A robust generalized bayes estimator and confidence region for a multivariate normal mean. Ann. Statist. 1980;8:716-761.

\bibitem{BhadraDattaPolsonWillard2017}%5
Bhadra~A, Datta~J, Polson~NG, Willard~B. The horseshoe+ estimator of ultra-sparse signals. Bayesian Anal. 2017;12:1105-1131.

\bibitem{BhattacharyaPatiPillaiDunson2015}%6
Bhattacharya~A, Pati~D, Pillai~NS, Dunson~DB. Dirichlet-laplace priors for optimal shrinkage. J. Amer. Statist. Assoc. 2015;110:1479-1490.

\bibitem{BogdanChakrabartiFrommletGhosh2011}%7
Bogdan~M, Chakrabarti~A, Frommlet~F, Ghosh~JK. Asymptotic bayes-optimality under sparsity of some multiple testing procedures. Ann. Statist. 2011;39:1551-1579.

\bibitem{CarvalhoPolsonScott2009}%8
Carvalho~CM, Polson~NG, Scott~JG. Handling sparsity via the horseshoe. Proceedings of the Twelfth International Conference on Artificial Intelligence and Statistics, PMLR. 2009;5:73-80.

\bibitem{CarvalhoPolsonScott2010}%9
Carvalho~CM, Polson~NG, Scott~JG. The horseshoe estimator for sparse signals. Biometrika. 2010;97:465-480.

\bibitem{CastilloRoquain2018}%10
Castillo~I, Roquain~E. On spike and slab empirical bayes multiple testing. arXiv pre-print arXiv: 1808.09748. 2018.

\bibitem{DattaGhosh2013}%11
Datta~J, Ghosh~JK. Asymptotic properties of bayes risk for the horseshoe prior. Bayesian Anal. 2013;8:111-132.

\bibitem{Efron2010}%12
Efron~B. The future of indirect evidence. Statist. Sci. 2010;25:145-157.

\bibitem{GhoshChakrabarti2017}%13
Ghosh~P, Chakrabarti~A. Asymptotic optimality of one-group shrinkage priors in sparse high-dimensional problems. Bayesian Anal. 2017;12:1133-1161.

\bibitem{GhoshTangGhoshChakrabarti2016}%14
Ghosh~P, Tang~X, Ghosh~M, Chakrabarti~A. Asymptotic properties of bayes risk of a general class of shrinkage priors in multiple hypothesis testing under sparsity. Bayesian Anal. 2016;11:753-796.

\bibitem{GriffinBrown2013}%15
Griffin~JE, Brown~PJ. Some priors for sparse regression modeling. Bayesian Anal. 2013;8:691-702.

\bibitem{JohnstoneSilverman2004}%16
Johnstone~IM, Silverman~BW. Needles and straw in haystacks: Empirical bayes estimates of possibly sparse sequences. Ann. Statist. 2004;32:1594-1649.

\bibitem{ParkCasella2008}%17
Park~T, Casella~G. The bayesian lasso. J. Amer. Statist. Assoc. 2008;103:681-686.

\bibitem{PolsonScott2012}%18
Polson~NG, Scott~JG. On the half-cauchy prior for a global scale parameter. Bayesian Anal. 2012;7:887-902.

\bibitem{Salomond2017}%19
Salomond~JB. Risk quantification for the thresholding rule for multiple testing using gaussian scale mixtures. arXiv pre-print arXiv: 1711.08705. 2017.

\bibitem{Singhetal2002}%20
Singh~D, Febbo~PG, Ross~K, Jackson~DG, Maonla~J, Ladd~C, Tamayo~P, Renshaw~AA, D'Amico~AV, Richie~JP, Lander~ES, Loda~M, Kantoff~PW, Golub~TR, Sellers~WR. Gene expression correlates of clinical prostate cancer behavior. Cancer Cell. 2002;2:203-209.

\bibitem{SongCheng2018}%21
Song~Q, Cheng~G. Optimal false discovery control of minimax estimator. arXiv pre-print arXiv: 1812.10013. 2018.

\bibitem{Strawderman1971}%22
Strawderman~WE. Proper bayes minimax estimators of the multivariate normal mean. Ann. Math. Statist. 1971;42:385-388.

\bibitem{VanDerPasSalomondSchmidtHieber2016}%23
van der Pas~S, Salomond~JB, Schmidt-Hieber~J. Conditions for posterior contraction in the sparse normal means problem. Eletron. J. Statist. 2016;10:976-1000. 

\bibitem{VanDerPasKleijnVanDerVaart2014}%24
van der Pas~SL, Kleijn~BJK, van der Vaart~AW. The horseshoe estimator: Posterior concentration around nearly black vectors. Electron. J. Statist. 2014;8:2585-2618.

\bibitem{VanDerPasSzaboVanDerVaart2017}%25
van der Pas~S, Szab\'{o}~B, van der Vaart~A. Adaptive posterior contraction rates for the horseshoe. Electron. J. Statist. 2017;11:3196-3225.

\bibitem{WellcomeTrust2007}%26
Wellcome Trust. Genome-wide association study of 14,000 cases of seven common diseases and 3000 shared controls. Nature. 2007;447:661-678.

\end{thebibliography}
\end{document}